\def\llncs{0}
\def\fullpage{1}
\def\anonymous{0}
\def\authnote{1}
\def\draft{0} 
\def\notxfont{0}
\def\submission{0} 

\ifnum\submission=1
\def\anonymous{1}
\def\llncs{1}
\def\fullpage{0}
\fi

\ifnum\llncs=1 	
\documentclass[envcountsect,a4paper,runningheads,11pt]{llncs}
\else
    \documentclass[letterpaper,hmargin=1.05in,vmargin=1.05in,11pt]{article}
    \ifnum\fullpage=1
        \usepackage{fullpage}
    \fi
\fi

\usepackage[%
  colorlinks=true,
  citecolor=blue,
  backref=page
]{hyperref}

\usepackage{amsmath, amsfonts, amssymb, mathtools,amscd}

\usepackage{amsthm}

\usepackage{lmodern}
\usepackage[T1]{fontenc}
\usepackage[utf8]{inputenc}

\usepackage{arydshln} 
\usepackage{url}
\usepackage{ifthen}
\usepackage{bm}
\usepackage{multirow}
\usepackage[dvips]{graphicx}
\usepackage[usenames]{color}
\usepackage{xcolor,colortbl} 
\usepackage{threeparttable}
\usepackage{comment}
\usepackage{paralist,verbatim}
\usepackage{cases}
\usepackage{booktabs}
\usepackage{breakcites}
\usepackage{braket}
\usepackage{cancel} 
\usepackage{ascmac} 
\usepackage{framed}
\usepackage{authblk}
\usepackage{pifont}
\usepackage{qcircuit}
\usepackage{tikz}
\usepackage{xspace} 
\usetikzlibrary{decorations.pathmorphing,decorations.shapes,decorations.markings}
\definecolor{darkblue}{rgb}{0,0,0.6}
\definecolor{darkgreen}{rgb}{0,0.5,0}
\definecolor{maroon}{rgb}{0.5,0.1,0.1}
\definecolor{dpurple}{rgb}{0.2,0,0.65}

\usepackage[capitalise,noabbrev]{cleveref}
\usepackage[absolute]{textpos}
\usepackage[final]{microtype}
\usepackage[absolute]{textpos}
\usepackage{everypage}
\DeclareMathAlphabet{\mathpzc}{OT1}{pzc}{m}{it}

\usepackage{algorithmic}
\usepackage{algorithm}
\usepackage{here}

\ifnum\draft=1
    \usepackage[color]{showkeys}
    \definecolor{refkey}{cmyk}{0,0,0,.25}
    \definecolor{labelkey}{cmyk}{0,0,0,.7}
    \usepackage{CJKutf8}
\else
\fi

\newtheoremstyle{thicktheorem}%
{\topsep}
{\topsep}
{\itshape}{}%
{\bfseries}%
{.}
{ }%
{\thmname{#1}\thmnumber{ #2}%
		\thmnote{ (#3)}%
}

\newtheoremstyle{remark}
{\topsep}
{\topsep}
	{}
	{}
	{}
	{.}
	{ }
	{\textit{\thmname{#1}}\thmnumber{ #2}
			\thmnote{ (#3)}%
	}

\ifnum\llncs=0
	\theoremstyle{thicktheorem}
	\newtheorem{theorem}{Theorem}[section]
	\newtheorem{lemma}[theorem]{Lemma}
	\newtheorem{corollary}[theorem]{Corollary}
	
	\newtheorem{definition}[theorem]{Definition}

	\theoremstyle{remark}
	
	\newtheorem{remark}[theorem]{Remark}

\else
\fi

\Crefname{MyClaim}{Claim}{Claims}

	\crefname{theorem}{Theorem}{Theorems}
	\crefname{assumption}{Assumption}{Assumptions}
	\crefname{construction}{Construction}{Constructions}
	\crefname{corollary}{Corollary}{Corollaries}
	\crefname{conjecture}{Conjecture}{Conjectures}
	\crefname{definition}{Definition}{Definitions}
	\crefname{exmaple}{Example}{Examples}
	\crefname{experiment}{Experiment}{Experiments}
	\crefname{counterexample}{Counterexample}{Counterexamples}
	\crefname{lemma}{Lemma}{Lemmata}
	\crefname{observation}{Observation}{Observations}
	\crefname{proposition}{Proposition}{Propositions}
	\crefname{remark}{Remark}{Remarks}
	\crefname{claim}{Claim}{Claims}
	\crefname{fact}{Fact}{Facts}
	\crefname{note}{Note}{Notes}

\ifnum\llncs=1
 \crefname{appendix}{App.}{Appendices}
 \crefname{section}{Sec.}{Sections}
\else
\fi

\ifnum\llncs=1
\renewcommand*{\backref}[1]{}
\else
	\renewcommand*{\backref}[1]{(Cited on page~#1.)}
	\ifnum\notxfont=1
	\else
		\usepackage{newtxtext}
	\fi
\fi
\ifnum\authnote=0  
\newcommand{\mor}[1]{}
\newcommand{\minki}[1]{}
\newcommand{\takashi}[1]{}
\newcommand{\xagawa}[1]{}

\else
\newcommand{\mor}[1]{$\ll$\textsf{\color{red} Tomoyuki: { #1}}$\gg$}
\newcommand{\takashi}[1]{$\ll$\textsf{\color{orange} Takashi: { #1}}$\gg$}

\newcommand{\xagawa}[1]{$\ll$\textsf{\color{magenta} Keita: { #1}}$\gg$}

\fi

\newcommand{\SD}{\mathsf{SD}} 

\newcommand{\puzz}{\mathsf{puzz}}
\newcommand{\ans}{\mathsf{ans}}

\newcommand{\Samp}{\algo{Samp}}

\newcommand{\cA}{\mathcal{A}}
\newcommand{\cB}{\mathcal{B}}

\newcommand{\cD}{\mathcal{D}}
\newcommand{\cE}{\mathcal{E}}

\newcommand{\cH}{\mathcal{H}}

\newcommand{\cO}{\mathcal{O}}

\newcommand{\cU}{\mathcal{U}}
\newcommand{\cV}{\mathcal{V}}

\def\makeuppercase#1{
\expandafter\newcommand\csname tl#1\endcsname{\widetilde{#1}}
}

\def\makelowercase#1{
\expandafter\newcommand\csname tl#1\endcsname{\widetilde{#1}}
}

\newcommand{\secp}{\lambda}

\newcommand*{\algo}[1]{\ensuremath{\mathsf{#1}}}

\newenvironment{boxfig}[2]{\begin{figure}[#1]\fbox{\begin{minipage}{0.97\linewidth}
                        \vspace{0.2em}
                        \makebox[0.025\linewidth]{}
                        \begin{minipage}{0.95\linewidth}
            {{
                        #2 }}
                        \end{minipage}
                        \vspace{0.2em}
                        \end{minipage}}}{\end{figure}}

\newcommand{\bit}{\{0,1\}}

\newcommand{\Gen}{\algo{Gen}}

\newcommand{\Ver}{\algo{Ver}}

\newcommand{\TD}{\algo{TD}}

\newcommand{\negl}{{\mathsf{negl}}}

\newcommand{\poly}{{\mathrm{poly}}}

\makeatletter
\DeclareRobustCommand
  \myvdots{\vbox{\baselineskip4\p@ \lineskiplimit\z@
    \hbox{.}\hbox{.}\hbox{.}}}
\makeatother

\title{Quantum Pessiland}

\ifnum\anonymous=1
\ifnum\llncs=1
\author{\empty}\institute{\empty}
\else
\author{}
\fi
\else
\ifnum\llncs=1
\author{
Tomoyuki Morimae\inst{1}
} 
\institute{
 Yukawa Institute for Theoretical Physics, Kyoto University, Kyoto, Japan 
}
\else
\author[1]{Boyang Chen}
\author[2]{Tomoyuki Morimae}
\author[3,2]{Takashi Yamakawa}
\affil[1]{{\small Department of Computer Science and Technology, Tsinghua University, Beijing, China}\authorcr{\small by-chen24@mails.tsinghua.edu.cn} }
\affil[2]{{\small Yukawa Institute for Theoretical Physics, Kyoto University, Kyoto, Japan}\authorcr{\small tomoyuki.morimae@yukawa.kyoto-u.ac.jp} }
\affil[3]{{\small NTT Social Informatics Laboratories, Tokyo, Japan}\authorcr{\small takashi.yamakawa@ntt.com} }
\fi 
\fi

\date{}

\begin{document}

\maketitle

\begin{abstract}
Pessiland is a world where $\mathsf{NP}$ is hard on average but one-way functions (OWFs) do not exist [Impagliazzo 1995].
Because almost all classical cryptographic primitives imply OWFs [Impagliazzo and Luby 1989], 
there is almost no classical cryptography in Pessiland.
On the other hand, 
quantum cryptography can exist even when OWFs do not [Kretschmer 2021; Morimae and Yamakawa 2022; Ananth, Qian and Yuen 2022].
Is there a quantum analogue of
Pessiland where $\mathsf{NP}$ is hard on average but even
quantum cryptography does not exist?   
In this paper, we show that such a miserable world, {\it Quantum Pessiland}, exists:
there is a quantum oracle relative to which $\mathsf{UP}\cap\mathsf{coUP}$ is hard on average against quantum polynomial-time algorithms
with quantum advice, yet
auxiliary-input EFI pairs do not exist.
We also show 
that there is a classical oracle relative to which $\mathsf{UP}\cap\mathsf{coUP}$ is hard on average against quantum polynomial-time algorithms
with quantum advice, yet
classically-secure auxiliary-input one-way puzzles (OWPuzzs) do not exist.
Almost all quantum cryptographic primitives imply EFI pairs or OWPuzzs, and therefore these results mean that
there is almost no quantum cryptography relative to these oracles.
We further show that relative to the classical oracle,
$\mathsf{SampBQP}=\mathsf{SampBPP}$, and therefore there is no sampling-based quantum advantage in Quantum Pessiland.
Finally, because our average-case hardness of $\mathsf{UP}\cap\mathsf{coUP}$ implies 
$\mathsf{P}^{\#\mathsf{P}}\not\subseteq\mathsf{i.o.BQP/qpoly}$,
our result also implies 
that a non-relativizing proof technique is necessary 
to construct OWPuzzs solely from
$\mathsf{P}^{\#\mathsf{P}}\not\subseteq\mathsf{i.o.BQP/qpoly}$,
which gives a partial negative answer to the open problem of [Khurana and Tomer 2025].
\end{abstract}

\ifnum\submission=0
\newpage
\setcounter{tocdepth}{2}
\tableofcontents
\newpage
\fi

\section{Introduction}

Pessiland~\cite{Impagliazzo95} is a world where $\mathsf{NP}$ is hard on average but one-way functions (OWFs) do not exist.
Because almost all classical cryptographic primitives imply OWFs~\cite{FOCS:ImpLub89}, 
Pessiland is a miserable world where essentially no classical cryptography exists even though $\mathsf{NP}$ is hard on average.

In quantum cryptography, by contrast, OWFs are not necessarily the minimal assumption~\cite{Kre21,C:MorYam22,C:AnaQiaYue22}.
Several quantum analogues of OWFs, pseudorandom generators (PRGs), and pseudorandom functions (PRFs) have been introduced, 
such as pseudorandom unitaries (PRUs)~\cite{C:JiLiuSon18}, pseudorandom state generators (PRSGs)~\cite{C:JiLiuSon18},
one-way state generators (OWSGs)~\cite{C:MorYam22}, one-way puzzles (OWPuzzs)~\cite{STOC:KhuTom24},
and EFI pairs~\cite{ITCS:BCQ23}.
They are potentially weaker than OWFs~\cite{Kre21,STOC:KQST23,STOC:KreQiaTal25,STOC:LomMaWri24}, yet
they still suffice for many applications, including private-key quantum money~\cite{C:JiLiuSon18}, secret-key encryption (SKE)~\cite{C:AnaQiaYue22},
digital signatures~\cite{C:MorYam22}, commitments~\cite{C:MorYam22,C:AnaQiaYue22,ITCS:BCQ23}, and multiparty computation~\cite{C:MorYam22,C:AnaQiaYue22}.
The fact that quantum cryptography may exist even without OWFs 
raises a natural question:
\begin{center}
\emph{Is there ``Quantum Pessiland'' where $\mathsf{NP}$ is hard on average but even
quantum cryptography does not exist?}    
\end{center}

\subsection{Main Results}
In this paper, we show that the answer is yes.
We prove the following two theorems.

\begin{theorem}
\label{thm:classicaloracle}
There exists a classical oracle $\cO$ relative to which the following statements hold.
\begin{enumerate}
    \item 
    There exists a language $L^\cO\in\mathsf{UP}^\cO\cap\mathsf{coUP}^\cO$ that is strongly hard on average
against quantum polynomial-time algorithms that can query $\cO$ and receive polynomial-size oracle-dependent quantum advice.\footnote{Because
$\mathsf{UP}\cap\mathsf{coUP}\subseteq\mathsf{NP}\cap\mathsf{coNP}$, the same result also holds for
$\mathsf{NP}\cap\mathsf{coNP}$.}
\item    
There exists a relation $R^\cO\in\mathsf{TFUP}^\cO$ that is  
average-case hard against quantum polynomial-time algorithms that can query $\cO$ and receive polynomial-size oracle-dependent quantum advice.
\item 
$\mathsf{SampBQP}^\cO=\mathsf{SampBPP}^\cO$.
(And therefore $\mathsf{BQP}^\cO=\mathsf{BPP}^\cO$.)
\item 
Classically-secure auxiliary-input OWPuzzs do not exist.
\end{enumerate}
\end{theorem}

\begin{theorem}
\label{thm:main-quantum-oracle}
There exists a quantum unitary oracle $\cU$ and
a classical oracle $\cV$ relative to which the following statements hold.
\begin{enumerate}
\item 
There exists a language $L^\cV\in\mathsf{UP}^\cV\cap\mathsf{coUP}^\cV$ that is strongly hard on average
against quantum polynomial-time algorithms that can query $(\cV,\cU)$ and receive polynomial-size oracle-dependent quantum advice.
\item    
There exists a relation $R^\cV\in\mathsf{TFUP}^\cV$ that is  
average-case hard against quantum polynomial-time algorithms that can query $(\cV,\cU)$ and receive polynomial-size oracle-dependent quantum advice.
\item 
Auxiliary-input EFI pairs do not exist relative to $(\cV,\cU)$.
\end{enumerate}
\end{theorem}

An EFI pair (generator)~\cite{ITCS:BCQ23} is a quantum polynomial-time (QPT) algorithm $\Gen(1^\secp,b)\to\rho_{\secp,b}$ that takes the security parameter $1^\secp$ and a bit $b\in\bit$ as input
and outputs a quantum state $\rho_{\secp,b}$ such that $\{\rho_{\secp,0}\}_{\secp\in\mathbb{N}}$ and $\{\rho_{\secp,1}\}_{\secp\in\mathbb{N}}$ are statistically far but computationally indistinguishable.
Almost all classical/quantum cryptographic primitives imply EFI pairs.
Moreover, EFI pairs are existentially equivalent to quantum non-interactive bit commitments
and multi-party computations~\cite{ITCS:BCQ23}.

Auxiliary-input EFI pairs
are a weaker variant of EFI pairs. 
An auxiliary-input EFI pair is a uniform QPT algorithm
$\mathsf{Gen}(x,b) \to \xi_{b,x}$ that takes a classical auxiliary input
$x \in \{0,1\}^*$ and a bit $b \in \{0,1\}$ as input and outputs a quantum
state $\xi_{b,x}$ such that there exists an infinite set
$S \subseteq \{0,1\}^*$ on which $\{\xi_{0,x}\}_{x \in S}$ and
$\{\xi_{1,x}\}_{x \in S}$ are statistically far but computationally
indistinguishable. (For the precise definition, see \cref{def:AIEFI}.)
Therefore, \cref{thm:main-quantum-oracle} also rules out the standard EFI pairs.
Almost all classical/quantum cryptographic primitives imply EFI pairs, and therefore 
non-existence of EFI pairs means that there is almost no classical/quantum cryptography.

OWPuzzs~\cite{STOC:KhuTom24} are a quantum analogue of OWFs. 
A OWPuzz is a pair $(\Samp,\Ver)$ of two algorithms.
$\Samp(1^\secp)\to(\puzz,\ans)$ is a QPT algorithm that takes the security parameter $1^\secp$ as input
and outputs two classical bit strings $\puzz$ and $\ans$.
$\Ver(\puzz,\ans')\to\top/\bot$ is a not-necessarily-efficient algorithm that takes $\puzz$ and a classical bit string $\ans'$ and
outputs $\top$ or $\bot$.
The correctness requires that 
\begin{align}
\Pr[\top\gets\Ver(\puzz,\ans):(\puzz,\ans)\gets\Samp(1^\secp)]\ge1-\negl(\secp)
\end{align}
and the security requires that
\begin{align}
\Pr[\top\gets\Ver(\puzz,\ans'):(\puzz,\ans)\gets\Samp(1^\secp),\ans'\gets\cA(1^\secp,\puzz)]\le\negl(\secp)
\end{align}
for any QPT adversary $\cA$.
We say that a OWPuzz is classically secure if the security is required to hold only against probabilistic polynomial-time (PPT) adversaries.
Almost all classical/quantum cryptographic primitives imply OWPuzzs, and OWPuzzs imply EFI pairs~\cite{STOC:KhuTom24}.

Auxiliary-input OWPuzzs are a weaker variant of OWPuzzs.
An auxiliary-input OWPuzz is a pair $(\mathsf{Samp},\mathsf{Ver})$ of
algorithms, where $\mathsf{Samp}(x)\to(\mathsf{puzz},\mathsf{ans})$ is a
uniform QPT algorithm that takes a classical auxiliary input
$x\in\{0,1\}^*$, and $\mathsf{Ver}(x,\puzz,\ans)\to \top/\bot$ is a not-necessarily-efficient algorithm.
The security is
required only on an infinite set of auxiliary inputs that may depend on
the adversary: for every adversary $\mathcal{A}$, there exists an
infinite set $I_{\mathcal{A}}\subseteq\{0,1\}^*$ such that 
for every $x\in I_{\mathcal{A}}$,
$\mathcal{A}$ finds a valid answer with only negligible probability.  
(For the precise definition, see \cref{def:auxiliary-input-owpuzz}.)
OWPuzzs imply classically-secure auxiliary-input OWPuzzs,
and therefore \cref{thm:classicaloracle} also rules out the standard OWPuzzs.
Most classical/quantum cryptographic primitives imply OWPuzzs.
Hence \cref{thm:classicaloracle} also rules out most classical/quantum cryptographic primitives.

\subsection{Implications}
Because the average-case hardness of $\mathsf{UP}\cap\mathsf{coUP}$ in \cref{thm:classicaloracle} implies that of $\mathsf{PP}$,
we also have the worst-case separation $\mathsf{PP}\not\subseteq\mathsf{i.o.BQP/qpoly}$.
Moreover, $\mathsf{PP}\subseteq\mathsf{P}^{\#\mathsf{P}}$.
Therefore \cref{thm:classicaloracle} gives the following corollary.
\begin{corollary}
\label{coro:KT}
There exists a classical oracle relative to which $\mathsf{PP}$ is strongly hard on average against QPT algorithms with quantum advice
(and therefore $\mathsf{P}^{\#\mathsf{P}}\not\subseteq\mathsf{i.o.BQP/qpoly}$) but
OWPuzzs do not exist.
\end{corollary}
Khurana and Tomer~\cite{STOC:KhuTom25} constructed OWPuzzs from 
$\mathsf{P}^{\#\mathsf{P}}\not\subseteq\mathsf{i.o.BQP/qpoly}$ 
and
an assumption that implies sampling-based quantum advantage (the so-called ``average-case  
$\mathsf{\#P}$ hardness assumption''~\cite{STOC:AarArk11,BreMonShe16}).
Whether OWPuzzs can be constructed solely from 
the average-case hardness of $\mathsf{PP}$ or
the worst-case separation $\mathsf{PP}\neq\mathsf{BQP}$ has remained open.
\cref{coro:KT} suggests that a non-relativizing technique would be required to resolve the open problem.

Morimae, Shirakawa, and Yamakawa~\cite{STOC:MorShiYam25} showed the equivalence between classically-secure OWPuzzs
and inefficient-verifier proofs of quantumness (IV-PoQ), which are interactive protocols demonstrating quantum 
advantage~\cite{C:MorYam24}.\footnote{
In \cite{STOC:MorShiYam25}, the security of OWPuzzs only requires $1-1/poly$ security 
(see their Remark 2.15 and 2.16). However, our proof breaks even such weak OWPuzzs.}
\cref{thm:classicaloracle} therefore also gives the following corollary,
which suggests that IV-PoQ cannot be constructed in a black-box way
from the average-case hardness of $\mathsf{UP}\cap\mathsf{coUP}$.
\begin{corollary}
There exists a classical oracle relative to which $\mathsf{UP}\cap\mathsf{coUP}$ is strongly hard on average against QPT algorithms with quantum advice,
but
IV-PoQ do not exist.
\end{corollary}

PoQ imply IV-PoQ, and therefore the above corollary also suggests that PoQ do not exist relative to the oracle.
No relation is known between IV-PoQ and $\mathsf{SampBQP}\neq\mathsf{SampBPP}$.
However, 
\cref{thm:classicaloracle} also shows 
$\mathsf{SampBQP}^\cO=\mathsf{SampBPP}^\cO$ relative to a classical oracle $\cO$.
Therefore, our results indicate that there is essentially no quantum advantage in Quantum Pessiland.

\subsection{Technical Overview}
Our starting point is the classical construction of Pessiland by
Wee~\cite{TCC:Wee06} (building on an unpublished result of Impagliazzo
and Rudich). Let $\pi=(\pi_1,\pi_2,\dots)$, where each $\pi_\ell$ is a
uniformly random permutation over $\{0,1\}^{\ell}$, and let $V^{\pi}$
be the \emph{verification} oracle
\begin{align}
V^{\pi}(\ell,u,v):=
\begin{cases}
1 & |u|=|v|=\ell \text{ and } \pi_\ell(u)=v,\\
0 & \text{otherwise}.
\end{cases}
\end{align}
Let QBF be an oracle that solves any PSPACE-complete problem.
Consider the oracle
\begin{align}
\cO\coloneqq (V^\pi,\operatorname{QBF}).    
\end{align}
\cite{TCC:Wee06} showed that, relative to $\cO$, there is a language in $\mathsf{NP}\cap\mathsf{coNP}$
(in fact $\mathsf{UP}\cap\mathsf{coUP}$) that is hard on average, but OWFs do not exist.

The result is shown as follows. 
First, the inverse $\pi^{-1}_n(y)$ is information-theoretically determined by
$V^{\pi}$ but hard to find: an algorithm making polynomially many
queries to $V^\pi$ essentially never receives the answer $1$ on a fresh query, so
the oracle it sees is indistinguishable from the trivial oracle $Z$ that always outputs zero. This
yields an average-case hard language in $\mathsf{NP}\cap\mathsf{coNP}$ (in fact
$\mathsf{UP}\cap\mathsf{coUP}$) via the Goldreich--Levin predicate
$\langle\pi^{-1}_n(y),r\rangle$. On the other hand, the same
observation shows that no candidate OWF
$f^{V^{\pi},\mathrm{QBF}}$ can effectively use $V^{\pi}$, and therefore 
$f^{V^\pi,\mathrm{QBF}}$ can be simulated by 
$f^{Z,\mathrm{QBF}}$ and the latter
can be
inverted in polynomial time with the help of the QBF oracle. 

Our goal is to carry
out this program in the quantum setting, where the candidate primitives
are quantum algorithms that query the oracles in superposition, where
the hardness must hold even against QPT algorithms with quantum advice, and where the
primitives to be ruled out are the \emph{auxiliary-input} versions of
OWPuzzs and EFI pairs.\footnote{\cite{TCC:Wee06} also ruled out auxiliary-input OWFs, but his technique cannot be used for our quantum setting. For details, see \cref{sec:relatedworks}.} The hardness side follows from a recent tight
bound for quantum permutation inversion with quantum
advice~\cite{EC:ABCGY26} combined with Goldreich--Levin with quantum
advice~\cite{STOC:KhuTom24,STACS:AdcCle02}. In this overview we focus on the
non-existence side, where the main technical novelty of this paper
lies.

\paragraph{The non-auxiliary-input case is easy.}
First, let us consider the easy case, namely, ruling out \emph{standard} (not
auxiliary-input) primitives. 
In this case, we can follow the template of \cite{TCC:Wee06} which we have explained above.
Let us consider OWPuzzs. To break OWPuzzs, we have only to break
a distributional OWPuzz, $\Samp^{V^\pi,\operatorname{QBF}}$.
The only
input of $\Samp$ is $1^{\lambda}$, which is
independent of $\pi$. 
Hence $V^{\pi}$ can simply be
replaced by a trivial oracle $Z$ that always outputs 0. More precisely,
we can easily show
\begin{align}
\mathbb{E}_{\pi}
\SD\left(
\Samp^{V^{\pi},\operatorname{QBF}}(1^\lambda),
\Samp^{Z,\operatorname{QBF}}(1^\lambda)
\right)
\le\frac{1}{p_1(\secp)}
\end{align}
for any large polynomial $p_1$.
Take any polynomial $r$.
Then, taking 
$p_1(\lambda)\coloneqq r(\lambda)\lambda^2$,
and applying Markov's inequality,
\begin{align}
\Pr_{\pi}
\left[
\SD\left(
\Samp^{V^{\pi},\operatorname{QBF}}(1^\lambda),
\Samp^{Z,\operatorname{QBF}}(1^\lambda)
\right)
\ge\frac{1}{r(\secp)}
\right]
\le r(\secp)\times \frac{1}{p_1(\secp)}=\frac{1}{\secp^2}.
\end{align}
Because the right-hand side is summable,
from the Borel--Cantelli
lemma, 
with probability one over $\pi$, 
\begin{align}
\SD\left(
\Samp^{V^{\pi},\operatorname{QBF}}(1^\lambda),
\Samp^{Z,\operatorname{QBF}}(1^\lambda)
\right)
\le\frac{1}{r(\secp)}
\label{to1}
\end{align}
for all sufficiently large $\secp$.
Fix any such $\pi$.
We can easily construct a PPT algorithm $\cA$ querying QBF 
such that
\begin{align}
\SD
\left(
(\puzz,\cA^{\operatorname{QBF}}(1^\secp,\puzz))_{\puzz\gets\Samp^{Z,\operatorname{QBF}}(1^\secp)},
(\puzz,\ans)_{(\puzz,\ans)\gets\Samp^{Z,\operatorname{QBF}}(1^\secp)}
\right)
\le\negl(\secp).    
\label{to2}
\end{align}
Combining \cref{to1,to2} breaks the distributional OWPuzzs.

\paragraph{Why auxiliary inputs break this argument.}
The above argument does not work for auxiliary-input OWPuzzs, because of the following two reasons.

First, the auxiliary input $x$ can depend on $\pi$, and therefore $x$ could contain $(\ell,u,v)$ such that
$V^\pi(\ell,u,v)=1$. In other words, $x$ can ``tell'' $\Samp$ some positive points of $V^\pi$. 
Then we can no longer replace $V^\pi$ with the trivial oracle $Z$.
To solve the issue, we have to find a ``learning'' QPT algorithm $\cE^{V^\pi,\operatorname{QBF}}$
that takes $x$ as input, 
runs $\Samp^{V^\pi,\operatorname{QBF}}(x)$ several times,
and outputs a ``patch'' $\Gamma$, which is a set of polynomially many $(\ell,u,v)$,
such that 
\begin{align}
\mathbb{E}_{\Gamma}
\SD\left(\Samp^{V^\pi,\operatorname{QBF}}(x),
\Samp^{\Gamma,\operatorname{QBF}}(x)\right)
\le\frac{1}{\poly(|x|)},
\label{to3}
\end{align}
where $\Samp^{\Gamma,\operatorname{QBF}}$ is the same as $\Samp^{V^\pi,\operatorname{QBF}}$ except that
when the base algorithm queries $(\ell,u,v)$ to $V^\pi$, it queries the oracle $\Gamma$ that
returns 1 if $(\ell,u,v)\in\Gamma$ and 0 otherwise.

Second, to break the auxiliary-input security,
for a fixed $\pi$,
\cref{to3} has to be satisfied \emph{for every sufficiently long} $x$.

\paragraph{The patching lemma.}
Our main technical tool, the patching lemma (\cref{lem:uniform-public-input-erasure-complete}), overcomes both
obstructions. It states that, with probability one over $\pi$, for
every uniform QPT algorithm $\cA^{V^{\pi},\operatorname{QBF}}(x)$ which takes $x\in\bit^n$ as input and outputs a quantum state, 
and every accuracy polynomial $r$, 
there is a uniform QPT extractor $\cE$ that receives $x$ 
and outputs a polynomial-length patch $\Gamma$ such that,
for all sufficiently large $n$ and \emph{simultaneously for every}
$x\in\{0,1\}^{n}$,
\begin{align}
\mathbb{E}_{\Gamma}\left[
\TD\left(\cA^{V^{\pi},\operatorname{QBF}}(x),\cA^{\Gamma,\operatorname{QBF}}(x)\right)\right]
\le\frac{1}{r(n)}.
\label{to:patching}
\end{align}

The extractor $\cE$ addresses the first obstruction in the natural way. First,
it classically queries $V^\pi$ for all $(\ell,u,v)$ such that $\ell\le m$ with $m=O(\log n)$,
and constructs the complete truth table $\Gamma_0$ of $V^\pi$ for all ``short'' $(u,v)$. Because $m=O(\log n)$, this is possible in polynomial time.
Next, $\cE$ runs $\cA^{\Gamma_0,\operatorname{QBF}}$ up to the $i$th query with random $i\in[T]$, where $T$ is the number of queries
to $V^\pi$ made by $\cA$,
and measures the query register to get $(\ell,u,v)$. $\cE$ queries $(\ell,u,v)$ to $V^\pi$,
and if the result is 1, $(\ell,u,v)$ is added to $\Gamma_0$. In other words, $\Gamma_1\coloneqq \Gamma_0\cup \{(\ell,u,v)\}$.
$\cE$ repeats this $R$ times with an appropriate polynomial $R$, 
and gets the list $(\Gamma_0,\Gamma_1,...,\Gamma_R)$.
Finally $\cE$ outputs $\Gamma\coloneqq\Gamma_J$ with random $J\gets\{0,1,...,R-1\}$.
A BBBV-style hybrid~\cite{BBBV97} shows that
\begin{align}
\mathbb{E}_{\Gamma}\left[\mathrm{TD}\left(\mathcal{A}^{V^{\pi},\operatorname{QBF}}(x),\mathcal{A}^{\Gamma,\operatorname{QBF}}(x)\right)\right]
\le 2T\sqrt{\frac{\kappa_{\pi}(x)}{R}},
\label{to4}
\end{align}
where 
$\kappa_{\pi}(x)$ is the expected number of stages at which the extractor discovers a new positive triple, i.e., adds an element to the patch. 

The key observation is that each stage makes a discovery with probability at most $1/R$, no matter what has happened before: conditioned on the entire transcript of the extraction algorithm $\cE$, the unrevealed part of each long block $\pi_{\ell}$ with $\ell\ge m$ is a uniformly random perfect matching consistent with the previously observed answers, and hence any fixed fresh triple is a true positive (i.e., a triple $(\ell,u,v)$ such that $V^\pi(\ell,u,v)=1$) with probability at most $1/(2^{m}-R)\le 1/R$, where the parameters are chosen so that $2^{m}\ge 2R$. Hence, $\mathbb{E}_{\pi}[\kappa_{\pi}(x)]\le R\times\frac{1}{R}=1$.
Averaging over $\pi$ on both sides of \cref{to4}, we have
\begin{align}
\mathbb{E}_\pi\mathbb{E}_{\Gamma}\left[\mathrm{TD}\left(\mathcal{A}^{V^{\pi},\operatorname{QBF}}(x),\mathcal{A}^{\Gamma,\operatorname{QBF}}(x)\right)\right]
\le \mathbb{E}_\pi2T\sqrt{\frac{\kappa_{\pi}(x)}{R}}
\le \frac{2T}{\sqrt{R}}.
\label{to4_5}
\end{align}
Hence by Markov's inequality, we obtain that
\begin{align}
\mathbb{E}_{\Gamma}\left[\mathrm{TD}\left(\mathcal{A}^{V^{\pi},\operatorname{QBF}}(x),\mathcal{A}^{\Gamma,\operatorname{QBF}}(x)\right)\right]
\le \frac{1}{\poly(n)}
\label{ot5}
\end{align}
except with probability $\delta$, which is an inverse-polynomial, over the choice of $\pi$. 

This is, however, far from sufficient: the security of auxiliary-input primitives only requires an infinite set of hard auxiliary inputs, so to break them we must succeed, for one fixed $\pi$, on all but finitely many $x$. Hence a union bound over the $2^{n}$ inputs of length $n$ would require 
$\delta$ to be $2^{-\omega(n)}$ for each $x$, which Markov's inequality cannot provide.

\paragraph{Solution.}
The exponential moment saves us. Since every stage makes a discovery with probability at most $1/R$ regardless of the transcript, the number $K$ of discoveries during the $R$ stages satisfies
\begin{align}
\mathbb{E}\left[2^{K}\right]\le\left(1+\frac{1}{R}\right)^{R}\le e,
\end{align}
where the expectation is over $\pi$ and the internal randomness of the extraction algorithm $\cE$, and, crucially, the bound is uniform in $x$. By Jensen's inequality the same bound holds with $K$ replaced by $\kappa_{\pi}(x)$, and bounding the maximum by the log-sum-exp,
\begin{align}
\mathbb{E}_{\pi}\left[\max_{x\in\{0,1\}^{n}}\kappa_{\pi}(x)\right]
\le\log_{2}\sum_{x\in\{0,1\}^{n}}\mathbb{E}_{\pi}\left[2^{\kappa_{\pi}(x)}\right]
\le\log_{2}\left(2^{n}e\right)\le n+2.
\end{align}
In words, even an adversarially chosen input cannot force the extractor $\cE$ to make more than roughly $n$ discoveries on average: each discovery is a $1/R$-unlikely event, so accumulating many more than $n$ of them has probability far below $2^{-n}$, and this is what beats the union bound over the $2^{n}$ inputs. 

Choosing $R$ to be a sufficiently large polynomial in $n$, $T$, and the accuracy polynomial $r$, we obtain $2T\sqrt{(n+2)/R}\le 1/(2r(n)n^{2})$, so that 
$\mathbb{E}_\pi\mathbb{E}_{\Gamma}\left[\mathrm{TD}\left(\mathcal{A}^{V^{\pi},\operatorname{QBF}}(x),\mathcal{A}^{\Gamma,\operatorname{QBF}}(x)\right)\right]$
is bounded by $1/2r(n)n^2$ even for the worst $x$. 
Markov's inequality over $\pi$ and the Borel--Cantelli lemma yield the final result
that \cref{to:patching} is true with probability one over $\pi$ and for all $x$.

\paragraph{Ruling out EFI pairs.}
Breaking EFI pairs can be done in a similar way. However, instead of the QBF oracle, we have to use
a quantum oracle, which we call the Helstrom oracle, because
it is not known whether classical oracle is enough to break EFI pairs~\cite{STOC:LomMaWri24}. 
The Helstrom oracle receives classical descriptions of two quantum circuits $C_{0}$ and $C_{1}$ together with a challenge state, 
and the oracle applies to the challenge state the optimal two-outcome (Helstrom) measurement that distinguishes the outputs 
of $C_{0}$ and $C_{1}$. (Actually, we consider a unitary version of it so that our oracle is a unitary oracle.)
Because the Helstrom measurement for two states generated by
polynomial-size circuits can be implemented in quantum polynomial space up to
an exponentially small error, 
the same proof goes through when the Helstrom oracle is replaced
with a $\mathsf{unitaryPSPACE}$ oracle.

\subsection{Related Works}
\label{sec:relatedworks}
\paragraph{Comparison with Aaronson and Chen~\cite{CCC:AarChe17}.}
Aaronson and Chen~\cite{CCC:AarChe17} constructed a classical oracle relative to which
$\mathsf{SampBQP} = \mathsf{SampBPP}$ and yet the polynomial-time hierarchy is infinite.
This implies that in order to show 
\begin{align}
\mbox{Polynomial-time hierarchy is infinite}~~\Rightarrow~~\mathsf{SampBQP}\neq\mathsf{SampBPP},
\end{align}
which is a long-standing open problem in the field of quantum advantage,
a non-relativizing technique is required.
(Currently what is known is that the infiniteness of the polynomial-time hierarchy plus
some ad hoc assumptions 
imply $\mathsf{SampBQP}\neq\mathsf{SampBPP}$~\cite{STOC:AarArk11,BreMonShe16}, and to remove the latter is an open problem.)

Our \cref{thm:classicaloracle} is of the
same flavor: relative to a classical oracle there is no sampling-based quantum
advantage, while strong average-case hardness survives. However, the two
results strengthen the hardness side along incomparable axes. The hardness
in \cite{CCC:AarChe17} is a worst-case structural statement about the entire hierarchy,
whereas ours concerns a single language: it is average-case rather than
worst-case, it holds against QPT adversaries with polynomial-size quantum
advice rather than against classical algorithms, and the hard language lies
in $\mathsf{UP} \cap \mathsf{coUP}$, a subclass of $\mathsf{NP}$. In these three respects our
hardness statement is stronger. Moreover, our oracle rules out not only sampling-based quantum advantage but also
inefficient-verifier proofs of quantumness. On the other hand, since our hard language
lies in $\mathsf{UP} \cap \mathsf{coUP}$, our result says nothing about separating the
levels of the hierarchy, and we do not know whether the polynomial-time
hierarchy is infinite relative to our oracle; 
in this respect their result is stronger, making the two results incomparable overall.

\paragraph{Comparison with Wee~\cite{TCC:Wee06}.}
The second result of Wee~\cite{TCC:Wee06} is the 
construction of a relativized Pessiland that rules out auxiliary-input OWFs,
and our patching lemma (\cref{lem:uniform-public-input-erasure-complete}) is similar in spirit to the approximation lemma
at the core of his proof. Relative to an oracle consisting of the verification oracle of a
uniformly random function $f$ together with a $\mathsf{PSPACE}$ oracle, Wee showed that
auxiliary-input OWFs do not exist, as follows. Since an auxiliary-input OWF is computed by a
family of polynomial-size oracle circuits, it suffices to show that every oracle circuit
$\mathcal{C}$ of polynomial size admits a small set $Q$ of inputs to $f$, depending on both
$\mathcal{C}$ and $f$, such that hardwiring the values of $f$ on $Q$ into $\mathcal{C}$ and
answering every query outside $Q$ by $0$ yields an oracle-free circuit that agrees with the
original one on all but a small fraction of inputs; the oracle-free circuit can then be
inverted with the help of the $\mathsf{PSPACE}$ oracle, and the inversion carries over to the
original circuit by the agreement. The set $Q$ is constructed iteratively over the query positions: for each $i$,
one adds to $Q$ all ``heavy'' points for the $i$-th query, namely those points
$x$ such that the $i$-th oracle query is of the form $(x,\cdot)$ on a
noticeable fraction of the circuit's inputs. The central difficulty is the same as ours: once the oracle
is fixed, the approximation must hold \emph{simultaneously} for all polynomial-size circuits,
and a per-circuit failure probability of $2^{-n}$, which is all that a naive replacement of
the oracle by the all-zero oracle achieves, cannot survive a union bound over the
$2^{\mathrm{poly}(n)}$ circuits. Wee resolved this by showing, via the Hoeffding bound, that
the light queries cause noticeable disagreement with probability only $e^{-\Omega(sn^2)}$
over $f$, which beats the number $2^{O(sn\log(sn))}$ of circuits of size $s$. This is
precisely the role played in our proof by the exponential moment bound on the number of
discoveries, which yields a failure probability small enough to beat the union bound over the
$2^{n}$ auxiliary inputs of each length.

Beyond this analogy, however, his technique does not transfer to our setting, for three
reasons. First, auxiliary-input OWFs are secure against nonuniform adversaries, so the set
$Q$ and the values of $f$ on $Q$ merely need to \emph{exist}: they are hardwired into the
inverter as nonuniform advice. In contrast, the security of auxiliary-input OWPuzzs and EFI
pairs is defined with respect to uniform adversaries, together with an adversary-dependent
infinite set of hard auxiliary inputs, so we must exhibit a single uniform QPT extractor that
\emph{efficiently learns} the patch by interacting with the oracle. Second, the heaviness of
a query is defined in terms of the fraction of the circuit's classical inputs on which the
query takes a given value, a quantity that is no longer meaningful for algorithms that query
the oracle in superposition; our extractor instead runs the algorithm up to a randomly chosen
query, measures the query register, and tests the outcome against $V^{\pi}$, in the spirit of
the heavy-query learners of Impagliazzo and Rudich~\cite{STOC:ImpRud89} and Barak and
Mahmoody~\cite{C:BarMah09}, combined with a BBBV-style hybrid argument~\cite{BBBV97}. Third,
the probabilistic structure is different: the values of a uniformly random function are
independent, which is what permits a direct application of the Hoeffding bound, whereas our
oracle encodes uniformly random permutations, so we instead argue that, conditioned on the
entire transcript of the extraction, every fresh triple tested by the extractor is a true
positive with probability at most the reciprocal of the number of extraction stages, and lift
this martingale-type bound to the exponential moment estimate. For these reasons, we view
the patching lemma as a uniform, quantum analogue of Wee's approximation lemma rather than an
application of it.
\section{Preliminaries}

\subsection{Basic Notations}
We use the standard notation of quantum information and cryptography.
The security parameter is $\secp$.
$[n]$ is the set $\{1,2,...,n\}$.
For a set $S$, $x\gets S$ means that an element $x$ is sampled uniformly at random from $S$.
For a bit string $x$, $|x|$ is the length of $x$.
For an algorithm $\cA$, $y\gets \cA(x)$ means that the algorithm $\cA$ is run on input $x$ and outputs $y$.
For simplicity, we often denote $\cA(x)$ by the output of the algorithm $\cA$ on input $x$. (If $\cA$ is a probabilistic algorithm that outputs a classical bit string,
$\cA(x)$ is a classical distribution. If $\cA$ outputs a quantum state, $\cA(x)$ is a quantum state.)
$\TD$ is the trace distance.
$\SD$ is the statistical distance.
$\negl$ is a negligible function.

\subsection{Average-Case Decision and Search Hardness}
Throughout this subsection, we fix an oracle $\mathcal{O}$.
For a polynomial $S$, an $S(N)$-qubit oracle-dependent quantum
advice family is a family of density operators
$\left\{
\rho_{N,\mathcal{O}}
\right\}_{N\in\mathbb{N}},
$
where $\rho_{N,\mathcal{O}}$ is an $S(N)$-qubit state depending on
the security parameter $N$ and on the oracle $\mathcal{O}$.\footnote{When an oracle $\cO$ is fixed, the advice states can
of course depend on the oracle. Here, our intention is to write the oracle dependence explicitly for the later use.}

\begin{definition}[Strong average-case hardness of languages]
\label{def:strong-average-hardness}
Let $L^{\mathcal{O}}\subseteq\{0,1\}^{*}$ be a language relative to
a fixed oracle $\mathcal{O}$. 
We say that $L^{\mathcal{O}}$ is strongly hard on average against
uniform QPT algorithms that can query $\cO$ if, for every
uniform QPT oracle algorithm $\mathcal{A}$, there exists a negligible
function $\negl$ such that
\begin{align}
\Pr_{z\leftarrow\{0,1\}^{N}}
\left[
\mathcal{A}^{\mathcal{O}}(1^{N},z)
=L^{\mathcal{O}}(z)
\right]
\leq
\frac{1}{2}+\negl(N).
\label{eq:strong-average-hardness}
\end{align}
We say that $L^{\mathcal{O}}$ is strongly hard on average against
QPT algorithms that can query $\cO$ and can receive polynomial-size oracle-dependent quantum
advice if, for every polynomial $S$, every uniform QPT oracle
algorithm $\mathcal{A}$, and every $S(N)$-qubit oracle-dependent 
quantum advice family
$\{\rho_{N,\mathcal{O}}\}_{N\in\mathbb{N}}$, there exists a negligible
function $\negl$ such that 
\begin{align}
\Pr_{z\leftarrow\{0,1\}^{N}}
\left[
\mathcal{A}^{\mathcal{O}}
(1^{N},z,\rho_{N,\mathcal{O}})
=L^{\mathcal{O}}(z)
\right]
\leq
\frac{1}{2}+\negl(N).
\label{eq:strong-average-hardness-quantum-advice}
\end{align}
\end{definition}

\begin{remark}
In the above definition, an instance $z$ is sampled uniformly at random. We could define      
the average-case hardness in a more general way where an instance is sampled from some QPT or PPT samplers, but
in this paper, we can show the hardness with the uniform distribution. 
\end{remark}

\begin{definition}[Average-case hardness of total unique search]
\label{def:average-case-hard-tfup}
Let
$R^{\mathcal{O}}
\subseteq
\{0,1\}^{*}\times\{0,1\}^{*}$
be a relation relative to a fixed oracle $\mathcal{O}$.
The relation is polynomially balanced if there exists a polynomial
$p$ such that
$|w|\leq p(|x|)$
for all 
$(x,w)\in R^{\mathcal{O}}$.
We say that
$R^{\mathcal{O}}\in\operatorname{TFUP}^{\mathcal{O}}$
if the following three properties are satisfied:
\begin{enumerate}
\item 
$R^{\mathcal{O}}$ is polynomially balanced.
\item 
There exists a
classical deterministic polynomial-time oracle algorithm $\mathcal{V}$ such
that, for every $x$ and $w$,
$\mathcal{V}^{\mathcal{O}}(x,w)
=
R^{\mathcal{O}}(x,w)$.
\item
For every instance $x$,
$|
\{
w\in\{0,1\}^{*}
:
R^{\mathcal{O}}(x,w)=1
\}
|
=
1.
$
In other words, every instance has exactly one polynomial-length valid solution.
\end{enumerate}
We say that $R^{\mathcal{O}}$ is average-case hard 
against uniform QPT algorithms that can query $\cO$
if, for every uniform QPT oracle algorithm $\mathcal{A}$,
there exists a negligible function $\negl$ such that
\begin{align}
\Pr
\left[
R^{\mathcal{O}}(x,\widehat{w})=1
:
x\leftarrow\bit^N,
\widehat{w}\leftarrow
\mathcal{A}^{\mathcal{O}}(1^{N},x)
\right]
\leq
\negl(N).
\label{eq:average-case-search-hardness}
\end{align}
We say that $R^{\mathcal{O}}$ is average-case hard 
against QPT algorithms that can query $\cO$ and can receive polynomial-size
oracle-dependent quantum advice if, for every polynomial $S$, every
uniform QPT oracle algorithm $\mathcal{A}$, and every $S(N)$-qubit
oracle-dependent quantum advice family
$\{\rho_{N,\mathcal{O}}\}_{N\in\mathbb{N}}$, there exists a negligible
function $\negl$ such that
\begin{align}
\Pr
\left[
R^{\mathcal{O}}(x,\widehat{w})=1
:
x\leftarrow\bit^N,
\widehat{w}\leftarrow
\mathcal{A}^{\mathcal{O}}
(1^{N},x,\rho_{N,\mathcal{O}})
\right]
\leq
\negl(N).
\label{eq:average-case-search-hardness-quantum-advice}
\end{align}
\end{definition}

\subsection{Auxiliary-Input Primitives}
\begin{definition}[Auxiliary-input OWPuzzs]
\label{def:auxiliary-input-owpuzz}
Fix an oracle $\mathcal{O}$.
An auxiliary-input 
one-way puzzle 
(OWPuzz)
relative to $\mathcal{O}$ is a pair
$(\Samp,\Ver)$ of oracle algorithms with the following
syntax.
\begin{enumerate}
    \item
    $\Samp^{\mathcal{O}}(x)\to (\mathsf{puzz},\mathsf{ans})$:
    It is a uniform QPT oracle
    algorithm that, on input a classical auxiliary input
    $x\in\{0,1\}^{*}$, outputs a pair of classical strings $(\puzz,\ans)$.
    \item
    $\Ver^{\mathcal{O}}(x,\mathsf{puzz},\mathsf{ans}')\to\top/\bot$:
    It is a not-necessarily-efficient oracle algorithm
    that halts on every input and outputs $\top/\bot$.
\end{enumerate}
We require the following
properties.

\paragraph{Correctness.}
There exists a negligible function $\negl$ such that,
for every $x\in\{0,1\}^{*}$,
\begin{align}
\Pr\left[
\top\gets
\Ver^{\mathcal{O}}(x,\mathsf{puzz},\mathsf{ans})
:
(\mathsf{puzz},\mathsf{ans})
\leftarrow
\Samp^{\mathcal{O}}(x)
\right]
\geq
1-\negl(|x|).
\label{eq:auxiliary-input-owpuzz-correctness}
\end{align}

\paragraph{Security.}
For every uniform QPT oracle adversary
$\mathcal{A}$, there exist an infinite set
$I_{\mathcal{A}}
\subseteq
\{0,1\}^{*}$
and a negligible function $\negl$ such that, for every
$x\in I_{\mathcal{A}}$,
\begin{align}
\Pr\left[
\top\gets\Ver^{\mathcal{O}}
(x,\mathsf{puzz},\widehat{\mathsf{ans}})
:
\begin{array}{r}
(\mathsf{puzz},\mathsf{ans})
\leftarrow
\Samp^{\mathcal{O}}(x)
\\
\widehat{\mathsf{ans}}
\leftarrow
\mathcal{A}^{\mathcal{O}}
(x,\mathsf{puzz})
\end{array}
\right]
\leq
\negl(|x|).
\label{eq:auxiliary-input-owpuzz-security}
\end{align}

An AI-OWPuzz is called classically secure if
the security condition is required only for uniform PPT
oracle adversaries $\mathcal{A}$.
\end{definition}

\begin{definition}[Auxiliary-input distributional OWPuzzs]\label{def:aux-dist-owpuzz}
Fix an oracle $\cO$. An auxiliary-input distributional OWPuzz
relative to $\cO$ is a uniform QPT oracle algorithm
$\mathsf{Samp}^{\cO}(x)\to(\mathsf{puzz},\mathsf{ans})$ that, on input a
classical auxiliary input $x\in\{0,1\}^{*}$, outputs a pair of classical
strings. 
We require the following security.
\paragraph{\bf Security.}
There exists a polynomial $p$ such that, for every uniform QPT oracle
adversary $\mathcal{A}$, there exists an infinite set
$I_{\mathcal{A}}\subseteq\{0,1\}^{*}$ such that, for every
$x\in I_{\mathcal{A}}$,
\begin{align}\label{eq:dist-security}
\mathrm{SD}
\left(
(\mathsf{puzz},\mathcal{A}^{\cO}(x,\mathsf{puzz}))_{
\mathsf{puzz}\gets\mathsf{Samp}^{\cO}(x)}
,
(\puzz,\ans)_{
(\puzz,\ans)\gets\mathsf{Samp}^{\cO}(x)}
\right)
\ge\frac{1}{p(|x|)}.
\end{align}

An auxiliary-input distributional OWPuzz is called \emph{classically secure}
if the security condition is required only for uniform PPT oracle
adversaries $\mathcal{A}$.
\end{definition}

The following lemma is straightforwardly obtained from the definitions.
\begin{lemma}[OWPuzzs imply distributional OWPuzzs, auxiliary-input version]
\label{lem:owpuzz-implies-dist}
Fix an oracle $\cO$. If $(\mathsf{Samp},\mathsf{Ver})$ is an auxiliary-input
(classically-secure) OWPuzz relative to $\cO$
then $\mathsf{Samp}$ is an auxiliary-input
(classically-secure) distributional OWPuzz relative to $\cO$. 
\end{lemma}

\begin{definition}[Auxiliary-input EFI pairs]\label{def:AIEFI}
Fix an oracle $\cO$. An auxiliary-input EFI pair (generator) relative to $\cO$ is a
uniform QPT oracle algorithm $\mathsf{Gen}$ that, on input a classical auxiliary
input $x\in\{0,1\}^{*}$ and a bit $b\in\{0,1\}$, outputs a state
$\xi_{b,x}\leftarrow\mathsf{Gen}^{\cO}(x,b)$. We require that there exists a positive
polynomial $p$ such that, for every uniform QPT oracle distinguisher $\cD$, there
exists an infinite set $S_{\cD}\subseteq\{0,1\}^{*}$ satisfying the following.

\textbf{Statistical farness.} For every $x\in S_{\cD}$,
\begin{align}
\mathrm{TD}\left(\xi_{0,x},\xi_{1,x}\right)\ge\frac{1}{p(|x|)}.
\label{eq:EFIfar}
\end{align}


\textbf{Computational indistinguishability.} 
There exists a negligible function $\negl$ such that for all $x\in S_{\cD}$,
\begin{align}
\left|\Pr\left[1\leftarrow \cD^{\cO}(x,\xi_{0,x})\right]
-\Pr\left[1\leftarrow \cD^{\cO}(x,\xi_{1,x})\right]\right|\le\negl(|x|).
\label{eq:EFIind}
\end{align}
\end{definition}

\begin{remark}
Auxiliary-input EFI pairs were defined in \cite{ITCS:BCQ23}. Their definition is
stated in a different way from the above one, 
but is actually equivalent to the above one.
(Their infinite set $S$ is independent of the distinguisher, and they require
that for every polynomial $q$ there exists $x\in S$ such that the advantage is at most $1/q(|x|)$.) 
\end{remark}

\subsection{Sampling Complexity}

\begin{definition}[Sampling Problems~\cite{Aar14,ITCS:ABK24}]
\label{def:Samplingproblems}
A (polynomially-bounded) sampling problem $S$ is a collection of probability distributions $\{D_x\}_{x\in\bit^*}$, 
where $D_x$ is a distribution over $\bit^{p(|x|)}$, for some fixed polynomial $p$.
\end{definition}

\begin{definition}[{\bf SampBPP} and {\bf SampBQP}~\cite{Aar14,ITCS:ABK24}] 
\label{def:SampBQP}
{\sf SampBPP} is the class of (polynomially-bounded) sampling problems 
$S=\{D_x\}_{x\in\bit^*}$ for which there exists a PPT algorithm $\cB$ such that for all $x$ and all $\epsilon>0$, 
\begin{align}
\SD(\cB(x,1^{\lfloor 1/\epsilon \rfloor}),D_x) \le\epsilon, 
\end{align}
where $\cB(x,1^{\lfloor 1/\epsilon\rfloor})$ is the output probability distribution of 
$\cB$ on input $(x, 1^{\lfloor 1/\epsilon\rfloor})$. 
{\sf SampBQP} is defined
the same way, except that $\cB$ is a QPT algorithm rather than a PPT one.
\end{definition}

\section{Oracle}
\label{subsec:oracle-construction-complete}

For every $\ell\in\mathbb{N}$, let
$\pi_{\ell}$ be a permutation over $\bit^\ell$ chosen uniformly at random
and write
$\pi
:=
\left(
\pi_{1},
\pi_{2},
\ldots
\right)$.
For a fixed realization
$\pi$, define the oracle $V^\pi$ by
\begin{align}
V^{\pi}(\ell,u,v)
\coloneqq
\begin{cases}
1,
&
|u|=|v|=\ell
\text{ and }
\pi_{\ell}(u)=v,
\\
0,
&
\text{otherwise}.
\end{cases}
\label{eq:graph-verification-oracle-complete}
\end{align}

We use $\operatorname{QBF}$ as a fixed PSPACE-complete Boolean oracle.
We also use the following quantum oracle, which is fixed independently
of $\pi$.

\begin{definition}[Helstrom oracle]
\label{def:stratified-helstrom-complete} 
We define a set
$\mathcal{H}
=
(
\mathcal{H}^{(0)},
\mathcal{H}^{(1)},
\mathcal{H}^{(2)},
\ldots
)
$
of quantum oracles as follows.
The level zero oracle $\cH^{(0)}$ acts as the identity.
A query to the level-$(j+1)$ oracle $\cH^{(j+1)}$ contains $1^{j+1}$, classical descriptions of two
quantum circuits $C_0$ and $C_1$, a classical string $z$, a state $|\psi\rangle$, and a one-qubit answer register.  
The circuits $C_0$ and
$C_1$ may query levels
$\mathcal{H}^{(0)},\ldots,\mathcal{H}^{(j)}$, but they may not query
$V^{\pi}$ or any other oracle depending on $\pi$.
Let $\sigma_{b,z}$ be the output density operator of $C_b(z)$, and let
$M_{C_0,C_1,z}$
be the projector onto the positive eigenspace
of 
$\sigma_{1,z}-\sigma_{0,z}$.
On the state $|\psi\rangle$ and the answer
register, the oracle applies
\begin{align}
U_{C_0,C_1,z}
:=
\left(
I-M_{C_0,C_1,z}
\right)
\otimes I
+
M_{C_0,C_1,z}
\otimes X.
\label{eq:stratified-helstrom-unitary-complete}
\end{align}
Invalid descriptions act as the identity.  
This defines $\mathcal{H}$ inductively because a query to the level-$(j+1)$ oracle $\cH^{(j+1)}$
refers only to already-defined lower levels.
Note that the unary level encoding implies that a polynomial-time algorithm on
$n$-bit inputs can access only levels $j=\operatorname{poly}(n)$.
\end{definition}

\section{Average-Case Hard Problems}
\label{subsec:quantum-advice-inversion-complete}

We use the following theorem.

\begin{theorem}[Theorem~1 of \cite{EC:ABCGY26}]
\label{thm:ABCGY26}
Let $\pi$ be a uniformly random permutation over $[M]$.
A quantum algorithm $\cA$ 
can make at most $T$ forward queries to $\pi$ and can make queries to a side oracle $w$, which is independent of $\pi$.
$\cA$ receives a uniformly random $y\gets[M]$ as the challenge input. 
$\cA$ can also receive
an $S$-qubit quantum advice state $\rho_{\pi,\cA,w}$ that may
depend arbitrarily on the permutation $\pi$,  
on the description of $\cA$, 
and on the side oracle $w$,
but it must be independent of the challenge $y$.
Finally, $\cA$ outputs $x$. 
Then 
\begin{align}
\mathbb{E}_{\pi}
\Pr
\left[
\pi(x)=y
:
y\gets[M],
x\gets\cA^{\pi,w}(y,\rho_{\pi,\cA,w})
\right]
\leq
C
\left(
\frac{ST}{M}
+
\frac{T^2}{M}
\right),
\label{eq:external-permutation-inversion-bound-complete}
\end{align}
where $C$ is a universal constant.
\end{theorem}

\begin{remark}
Theorem~1 of \cite{EC:ABCGY26} is stated without the side oracle $w$.
Granting access to a side oracle $w$ whose action is fixed independently of $\pi$
is without loss of generality: the algorithm is constrained only in the number of
queries to $\pi$ and may apply arbitrary unitaries between the queries, so the
action of $w$ (and of its inverse) on the relevant registers can be absorbed into
these unitaries, and the dependence of the advice state on $w$ is subsumed by its
dependence on the description of $\cA$. We nevertheless state the side oracle
explicitly for our later uses.
\end{remark}

From \cref{thm:ABCGY26}, we obtain the following corollary.

\begin{corollary}
\label{cor:fixed-oracle-inversion-hardness}
Let
$\pi=(\pi_{1},\pi_{2},\ldots)$,
where, independently for every $\ell\in\mathbb{N}$,
$\pi_{\ell}$ is chosen uniformly from the permutations of
$\{0,1\}^{\ell}$.
Fix any collection $\mathcal{B}$ of classical or quantum side oracles
whose action is fixed independently of $\pi$.
Then, with probability one over the choice of $\pi$, the
following statement holds.

For every uniform QPT oracle algorithm $\mathcal{A}$, every pair of
polynomials $S$ and $T$ such that $\mathcal{A}$ makes at most $T(n)$
queries to $V^{\pi}$ on inputs of length $n$, and every family of
$S(n)$-qubit advice states
$\{\rho_{n,\pi,\mathcal{A},\mathcal{B}}\}_{n\in\mathbb{N}}$, there
exists a negligible function $\negl$ such that
\begin{align}
\Pr
\left[
 x=\pi_{n}^{-1}(y)
:
 y\leftarrow\{0,1\}^{n},
x\leftarrow
 \mathcal{A}^{V^{\pi},\mathcal{B}}
 \left(1^{n},y,\rho_{n,\pi,\mathcal{A},\mathcal{B}}\right)
\right]
\leq
\negl(n).
\label{eq:fixed-oracle-inversion-hardness}
\end{align}
The advice state may depend arbitrarily on the entire fixed
permutation family $\pi$, on the description of $\mathcal{A}$, and on
$\mathcal{B}$, but it must be independent of
$y$.
\end{corollary}

\begin{proof}
Fix a uniform QPT oracle algorithm $\mathcal{A}$ and polynomials $S$
and $T$ as in the statement.  We first prove a probability-one
statement for this fixed triple $(\mathcal{A},S,T)$ and then take a
countable intersection over all such triples.

Fix an input length $n$.  
Let
$\pi_{\neq n}\coloneqq(\pi_{\ell})_{\ell\neq n}$ denote a fixed realization of all
non-target permutation blocks, where each
$\pi_{\ell}$ is a permutation of $\{0,1\}^{\ell}$.
For an $S(n)$-qubit state $\rho$, define
\begin{align}
\operatorname{Succ}_{n}^{\mathcal{A},\mathcal{B}}
\left(
\pi_n,\pi_{\neq n};\rho
\right)
\coloneqq
\Pr
\left[
 x=\pi_n^{-1}(y)
 :
 y\leftarrow\{0,1\}^{n},
 x\leftarrow
 \mathcal{A}^{V^{\pi},\mathcal{B}}
 \left(1^{n},y,\rho\right)
\right].
\label{eq:conditioned-inversion-success-probability}
\end{align}
Define the best success probability obtainable from an $S(n)$-qubit
advice state by
\begin{align}
\operatorname{Opt}_{n}^{\mathcal{A},S,\mathcal{B}}
\left(
\pi_n,\pi_{\neq n}
\right)
\coloneqq
\max_{\rho
}
\operatorname{Succ}_{n}^{\mathcal{A},\mathcal{B}}
\left(
\pi_n,\pi_{\neq n};\rho
\right).
\label{eq:conditioned-optimal-inversion-success}
\end{align}

We next bound the expectation of
Equation~\eqref{eq:conditioned-optimal-inversion-success} over a
uniformly random $\pi_n$, while keeping $\pi_{\neq n}$ fixed.
A query to 
$V^{\pi}$ can be simulated by
at most two forward queries to $\pi_n$.
Applying \cref{thm:ABCGY26},
for every fixed
$\pi_{\neq n}$,
\begin{align}
\mathbb{E}_{\pi_n}
\left[
\operatorname{Opt}_{n}^{\mathcal{A},S,\mathcal{B}}
\left(
\pi_n,\pi_{\neq n}
\right)
\right]
\leq
C
\left(
\frac{2S(n)T(n)}{2^{n}}
+
\frac{4T(n)^{2}}{2^{n}}
\right).
\label{eq:conditioned-optimal-success-bound}
\end{align} 
Here, we have used the fact that a query to $V^\pi$ can be simulated by querying $\pi$ twice.
Hence the tower
property of conditional expectation gives
\begin{align}
\mathbb{E}_{\pi}
\left[
\operatorname{Opt}_{n}^{\mathcal{A},S,\mathcal{B}}
\left(
\pi_n,
\pi_{\neq n}
\right)
\right]
=
\mathbb{E}_{\pi_{\neq n}}
\left[
\mathbb{E}_{\pi_n}
\left[
\operatorname{Opt}_{n}^{\mathcal{A},S,\mathcal{B}}
\left(
\pi_n,
\pi_{\neq n}
\right)
\middle|
\pi_{\neq n}
\right]
\right]
\leq
C
\left(
\frac{2S(n)T(n)}{2^{n}}
+
\frac{4T(n)^{2}}{2^{n}}
\right).
\label{eq:unconditioned-optimal-success-bound}
\end{align}
Set
$\eta(n)
:=
2^{-n/4}$
and define the bad event
\begin{align}
\mathsf{Bad}_{n}^{\mathcal{A},S,T,\mathcal{B}}
:=
\left\{
\pi:
\operatorname{Opt}_{n}^{\mathcal{A},S,\mathcal{B}}
\left(
\pi_{n},
\pi_{\neq n}
\right)
>
\eta(n)
\right\}.
\label{eq:bad-event-in-corollary}
\end{align}
By Markov's inequality and
Equation~\eqref{eq:unconditioned-optimal-success-bound},
\begin{align}
\Pr_{\pi}
\left[
\mathsf{Bad}_{n}^{\mathcal{A},S,T,\mathcal{B}}
\right]
\leq
C
\left(
2S(n)T(n)
+
4T(n)^{2}
\right)
2^{-3n/4}.
\label{eq:markov-bound-in-corollary}
\end{align}
Since $S$ and $T$ are polynomials, the right-hand side of
Equation~\eqref{eq:markov-bound-in-corollary} is summable in $n$.
The first Borel--Cantelli lemma therefore implies that, for the fixed
triple $(\mathcal{A},S,T)$, with probability one over
$\pi$ only finitely many of the events
$\mathsf{Bad}_{n}^{\mathcal{A},S,T,\mathcal{B}}$ occur.
Equivalently, with probability one over $\pi$ there exists an integer
$n_{0}=n_{0}(\mathcal{A},S,T,\mathcal{B},\pi)$ such that
\begin{align}
\operatorname{Opt}_{n}^{\mathcal{A},S,\mathcal{B}}
\left(
\pi_n,
\pi_{\neq n}
\right)
\leq
2^{-n/4}
\label{eq:eventual-optimal-success-bound}
\end{align}
for every $n\geq n_{0}$.

There are only countably many uniform oracle machines and only
countably many polynomial bounds with finite descriptions.  Taking the countable intersection of these events of probability one
produces a set of complete permutation families that also has
probability one and for which
Equation~\eqref{eq:eventual-optimal-success-bound} holds
simultaneously for every triple $(\mathcal{A},S,T)$.
\end{proof}

\begin{lemma}[Goldreich-Levin with quantum advice~\cite{STOC:KhuTom24,STACS:AdcCle02}]
\label{lem:KT}
There exists a QPT algorithm $\mathcal{E}$ with the following property.
For every $n\in\mathbb{N}$, every QPT algorithm $\mathcal{D}$ that
outputs a bit, every mixed state $\rho$, every string
$a\in\{0,1\}^{n}$, and every $\varepsilon>0$, if
\begin{align}
\Pr_{r\leftarrow\{0,1\}^{n}}
\left[
\mathcal{D}(\rho,r)
=
\langle a,r\rangle
\right]
&\geq
\frac{1}{2}+\varepsilon,
\label{eq:kt-gl-prediction}
\end{align}
then
\begin{align}
\Pr
\left[
\mathcal{E}^{\mathcal{D}}(\rho)=a
\right]
&\geq
4\varepsilon^{2}.
\label{eq:kt-gl-recovery}
\end{align}
Here, $\mathcal{E}^{\mathcal{D}}$ denotes the coherent black-box use of a
unitary implementation of $\mathcal{D}$ and its inverse.
\end{lemma}

\begin{corollary}
\label{coro:averageKT}
Fix $n\in\mathbb{N}$ and an arbitrary family of strings
$\{ x_{y}\in\{0,1\}^{n} \}_{y\in\{0,1\}^{n}}$.
Let $\mathcal{O}$ be any fixed collection of reversible oracle
unitaries whose inverse queries are available, and let
$\mathcal{P}^{\mathcal{O}}$ be a uniform QPT oracle algorithm that
outputs a bit.  Let $\rho$ be an arbitrary mixed advice state that 
may depend on
$n$, on $\mathcal{O}$, on the description of $\mathcal{P}$, and on the
family $\{x_y\}_y$, but it may
not depend on $y$ and $r$.
Suppose that, for some $\varepsilon>0$,
\begin{align}
\Pr_{\substack{
 y\leftarrow\{0,1\}^{n}\\
 r\leftarrow\{0,1\}^{n}
}}
\left[
 \mathcal{P}^{\mathcal{O}}(\rho,y,r)
 =
 \langle x_{y},r\rangle
\right]
&\geq
\frac{1}{2}+\varepsilon.
\label{eq:average-index-prediction}
\end{align}
Then there exists a uniform QPT oracle algorithm
$\mathcal{R}^{\mathcal{O}}$, whose code depends only on the code of
$\mathcal{P}$, such that
\begin{align}
\Pr_{
 y\gets\{0,1\}^{n}
}
\left[
 x_y\gets \mathcal{R}^{\mathcal{O}}(\rho,y)
\right]
&\geq
4\varepsilon^{2}.
\label{eq:average-index-recovery}
\end{align}
\end{corollary}

\begin{proof}
For every $y\in\{0,1\}^{n}$, define 
\begin{align}
\delta_{y}
&:=
\Pr_{r\leftarrow\{0,1\}^{n}}
\left[
 \mathcal{P}^{\mathcal{O}}(\rho,y,r)
 =
 \langle x_{y},r\rangle
\right]
-\frac{1}{2},
\label{eq:conditional-signed-advantage}
\\
\delta_{y}^{+}
&:=
\max\{\delta_{y},0\}.
\label{eq:positive-part-advantage}
\end{align}
Equation~\eqref{eq:average-index-prediction} is equivalent to
\begin{align}
\mathbb{E}_{y\leftarrow\{0,1\}^{n}}
\left[
 \delta_{y}
\right]
&\geq
\varepsilon.
\label{eq:average-signed-advantage}
\end{align}
Fix $y\in\{0,1\}^{n}$ and define the predictor
$\mathcal{D}_{y}^{\mathcal{O}}$ by hardwiring $y$ into
$\mathcal{P}^{\mathcal{O}}$:
\begin{align}
\mathcal{D}_{y}^{\mathcal{O}}(\rho,r)
&:=
\mathcal{P}^{\mathcal{O}}(\rho,y,r).
\label{eq:fixed-y-predictor}
\end{align}
Define the recovery algorithm by
\begin{align}
\mathcal{R}^{\mathcal{O}}(\rho,y)
&:=
\mathcal{E}^{\mathcal{D}_{y}^{\mathcal{O}}}(\rho),
\label{eq:average-index-recovery-algorithm}
\end{align}
where $\mathcal{E}$ is defined
in
\cref{eq:kt-gl-recovery}. 
If $\delta_{y}>0$, then \cref{lem:KT}
applied with $a=x_{y}$ and
$\varepsilon=\delta_{y}$, gives
\begin{align}
\Pr
\left[
 \mathcal{R}^{\mathcal{O}}(\rho,y)=x_{y}
\right]
&\geq
4\delta_{y}^{2}.
\label{eq:fixed-y-positive-recovery}
\end{align}
If $\delta_{y}\leq 0$, the recovery probability is nonnegative.
Consequently, for every $y\in\{0,1\}^{n}$,
\begin{align}
\Pr
\left[
 \mathcal{R}^{\mathcal{O}}(\rho,y)=x_{y}
\right]
&\geq
4\left(\delta_{y}^{+}\right)^{2}.
\label{eq:fixed-y-recovery-positive-part}
\end{align}
Averaging Equation~\eqref{eq:fixed-y-recovery-positive-part} over a
uniform $y$ and applying Jensen's inequality to the convex function
$t\mapsto t^{2}$ yields
\begin{align}
\Pr_{\substack{
 y\leftarrow\{0,1\}^{n}\\
 \widehat{x}\leftarrow
 \mathcal{R}^{\mathcal{O}}(\rho,y)
}}
\left[
 \widehat{x}=x_{y}
\right]
\geq
4\mathbb{E}_{y}
\left[
 \left(\delta_{y}^{+}\right)^{2}
\right]
\geq
4\left(
 \mathbb{E}_{y}
 \left[
  \delta_{y}^{+}
 \right]
\right)^{2}
\geq
4\left(
 \mathbb{E}_{y}
 \left[
  \delta_{y}
 \right]
\right)^{2}
\geq
4\varepsilon^{2}.
\label{eq:average-recovery-final}
\end{align}
Here 
the final
inequality follows from
\cref{eq:average-signed-advantage}.
\end{proof}

\begin{lemma}[Decision and search hardness]
\label{lem:decision-and-search-hardness}
Let $\pi\coloneqq(\pi_1,\pi_2,...)$, where $\pi_\ell$ is a permutation over $\bit^\ell$.
For a string $z\in\{0,1\}^{N}$ with $N\geq2$, set
$n:=\lfloor N/2\rfloor$ and parse
$z
=
y\|r\|s$,
$|y|=|r|=n$,
and $|s|=N-2n$.
Define
the language
\begin{align}
L^{\pi}(z)
:=
\left\langle
\pi_n^{-1}(y),r
\right\rangle
\bmod 2.
\label{eq:hard-language-definition-complete}
\end{align}
In other words, $z \in L^\pi$ if and only if $\left\langle \pi_n^{-1}(y),r \right\rangle = 1$.
The three strings of length less than two, namely, the empty string, 0 and 1, may be assigned
arbitrary values.  Define also
the relation
\begin{align}
R^{\pi}(y,w)=1
\quad\Longleftrightarrow\quad
|y|=|w|
\text{ and }
V^{\pi}(|y|,w,y)=1.
\label{eq:hard-search-relation-complete}
\end{align} 
Let $\mathcal{B}$ be any fixed collection of reversible oracle
unitaries independent of $\pi$, and suppose that inverse
queries to every oracle in $\mathcal{B}$ are available.  Let $\pi$ be
a permutation family satisfying
Corollary~\ref{cor:fixed-oracle-inversion-hardness}.
Then, relative to $V^{\pi}$,
\begin{align}
R^{\pi}
&\in
\operatorname{TFUP}^{V^{\pi}},
\label{eq:hard-relation-in-tfup}
\\
L^{\pi}
&\in
\operatorname{UP}^{V^{\pi}}
\cap
\operatorname{coUP}^{V^{\pi}}.
\label{eq:hard-language-in-up-intersection-coup}
\end{align}
Moreover, relative to $(V^\pi,\cB)$, $R^{\pi}$ is average-case hard against 
QPT algorithms with polynomial-size oracle-dependent quantum advice,
and $L^{\pi}$ is strongly hard on average against
QPT algorithms with polynomial-size oracle-dependent quantum advice.
\end{lemma}

\begin{proof}
For every $y\in\{0,1\}^{n}$, the relation
$R^{\pi}$ has exactly one valid solution, namely
$w
=
\pi_{n}^{-1}(y).$
Every valid solution satisfies $|w|=|y|$, so the relation is
polynomially balanced.  A deterministic polynomial-time oracle
verifier checks whether
$V^{\pi}(n,w,y)
=
1$. 
This proves Equation~\eqref{eq:hard-relation-in-tfup}.  Furthermore,
the average-case
hardness of $R^{\pi}$ follows directly from
Corollary~\ref{cor:fixed-oracle-inversion-hardness}.

We next prove Equation~\eqref{eq:hard-language-in-up-intersection-coup}.
The unique witness for
both $L^{\pi}$ and its complement is
$w=\pi_{n}^{-1}(y)$.  The verifier first checks
$V^{\pi}(n,w,y)
=
1$,
and then checks
$\langle w,r\rangle
=
1
\pmod 2$
for membership in $L^{\pi}$, or
$
\langle w,r\rangle
=
0
\pmod 2$
for membership in its complement.  Since $\pi_n$ is a permutation,
there is exactly one possible witness. 
This
proves Equation~\eqref{eq:hard-language-in-up-intersection-coup}.

It remains to prove the strong average-case hardness of $L^{\pi}$. 
Assume toward a contradiction that it fails.  Then there exist a
polynomial $S$, a uniform QPT oracle algorithm $\mathcal{A}$, an
$S(N)$-qubit oracle-dependent quantum advice family
$\left\{
\rho_{N}
\right\}_{N\in\mathbb{N}}$\footnote{For notational simplicity, we omit the dependence of $\rho$ on $V^\pi$, $\cB$ and $\cA$.},
and a positive polynomial $q$ such that, for infinitely many input
lengths $N$,
\begin{align}
\Pr_{z\leftarrow\{0,1\}^{N}}
\left[
\mathcal{A}^{V^{\pi},\mathcal{B}}
\left(1^{N},z,\rho_{N}\right)
=
L^{\pi}(z)
\right]
&\geq
\frac{1}{2}
+
\frac{1}{q(N)}.
\label{eq:decision-nonnegligible-advantage}
\end{align}

For $n\in\mathbb{N}$ and $\beta\in\{0,1\}$, define
$N_{\beta}(n)
\coloneqq
2n+\beta$.
We define two uniform QPT predictors.  On input $(\rho,y,r)$ with
$|y|=|r|=n$, let
\begin{align}
\mathcal{P}_{0}^{V^{\pi},\mathcal{B}}(\rho,y,r)
&:=
\mathcal{A}^{V^{\pi},\mathcal{B}}
\left(1^{2n},y\|r,\rho\right).
\label{eq:even-length-predictor}
\end{align}
The predictor $\mathcal{P}_{1}$ samples
$s\leftarrow\{0,1\}$ and outputs
\begin{align}
\mathcal{P}_{1}^{V^{\pi},\mathcal{B}}(\rho,y,r)
&:=
\mathcal{A}^{V^{\pi},\mathcal{B}}
\left(1^{2n+1},y\|r\|s,\rho\right).
\label{eq:odd-length-predictor}
\end{align}
The random bit in Equation~\eqref{eq:odd-length-predictor} can be
generated coherently and its measurement can be deferred.  Since
$V^{\pi}$ is self-inverse and inverse queries to $\mathcal{B}$ are
available, both predictors admit the coherent forward and inverse
implementations required by \cref{coro:averageKT}.

For each $\beta\in\{0,1\}$, define 
\begin{align}
\alpha_{n,\beta}
:={}&
\Pr_{\substack{
y\leftarrow\{0,1\}^{n}\\
r\leftarrow\{0,1\}^{n}
}}
\left[
\mathcal{P}_{\beta}^{V^{\pi},\mathcal{B}}
\left(
\rho_{N_{\beta}(n)},y,r
\right)
=
\left\langle
\pi_{n}^{-1}(y),r
\right\rangle
\right]
-
\frac{1}{2}.
\label{eq:associated-prediction-advantage}
\end{align}
Because $L^{\pi}$ ignores the padding string $s$, 
\begin{align}
\alpha_{n,\beta}
={}&
\Pr_{z\leftarrow\{0,1\}^{N_{\beta}(n)}}
\left[
\mathcal{A}^{V^{\pi},\mathcal{B}}
\left(
1^{N_{\beta}(n)},
z,
\rho_{N_{\beta}(n)}
\right)
=
L^{\pi}(z)
\right]
-
\frac{1}{2}.
\label{eq:prediction-equals-decision-advantage}
\end{align}
For each $\beta\in\{0,1\}$, apply
\cref{coro:averageKT} to
$\mathcal{P}_{\beta}$.  This gives a uniform QPT recovery algorithm
$\mathcal{R}_{\beta}^{V^{\pi},\mathcal{B}}$ such that, whenever
$\alpha_{n,\beta}>0$,
\begin{align}
\Pr_{\substack{
y\leftarrow\{0,1\}^{n}\\
\widehat{w}_{\beta}
\leftarrow
\mathcal{R}_{\beta}^{V^{\pi},\mathcal{B}}
\left(
\rho_{N_{\beta}(n)},y
\right)
}}
\left[
\widehat{w}_{\beta}
=
\pi_{n}^{-1}(y)
\right]
&\geq
4\alpha_{n,\beta}^{2}.
\label{eq:recovery-for-each-parity}
\end{align}
Each invocation in Equation~\eqref{eq:recovery-for-each-parity}
uses one copy of the corresponding advice state.

Let $\mathcal{N}\subseteq\mathbb{N}$ be the infinite set of input
lengths for which
Equation~\eqref{eq:decision-nonnegligible-advantage} holds.  For
$\beta\in\{0,1\}$, define
$\mathcal{N}_{\beta}
\coloneqq
\left\{
 n\in\mathbb{N}
 :
 2n+\beta\in\mathcal{N}
\right\}.$
Since $\mathcal{N}$ is infinite, at least one of
$\mathcal{N}_{0}$ and $\mathcal{N}_{1}$ is infinite.  Fix
$\beta^{\star}\in\{0,1\}$ such that
$\mathcal{N}_{\beta^{\star}}$ is infinite.
For every $n\in\mathcal{N}_{\beta^{\star}}$,
\cref{eq:decision-nonnegligible-advantage} and
\cref{eq:prediction-equals-decision-advantage} give
$\alpha_{n,\beta^{\star}}
\geq
\frac{1}{q(2n+\beta^{\star})}.$
In particular, $\alpha_{n,\beta^{\star}}>0$, and hence
Equation~\eqref{eq:recovery-for-each-parity} implies
\begin{align}
\Pr_{\substack{
 y\leftarrow\{0,1\}^{n}\\
 \widehat{w}
 \leftarrow
 \mathcal{R}_{\beta^{\star}}^{V^{\pi},\mathcal{B}}
 \left(
 \rho_{2n+\beta^{\star}},y
 \right)
}}
\left[
 \widehat{w}
 =
 \pi_{n}^{-1}(y)
\right]
&\geq
\frac{4}{q(2n+\beta^{\star})^{2}}
\label{eq:fixed-parity-inversion-success}
\end{align}
for infinitely many $n$.
This contradicts
Corollary~\ref{cor:fixed-oracle-inversion-hardness}. 
\end{proof}

\section{Patching Lemma}
\label{subsec:adaptive-erasure-complete}

\begin{definition}[Patches]
\label{def:classical-patch-complete}
A \emph{patch} is a finite set $\Gamma$ of triples $(\ell,u,v)$ with
$|u|=|v|=\ell$.  It defines the oracle $\Gamma$ as
\begin{align}
\Gamma(\ell,u,v)
:=
\begin{cases}
1, & (\ell,u,v)\in\Gamma,\\
0, & \text{otherwise}.
\end{cases}
\label{eq:patched-boolean-oracle-complete}
\end{align}
\end{definition}


\begin{lemma}[Patching lemma]
\label{lem:uniform-public-input-erasure-complete}
Fix a side oracle $\mathcal{B}$ independent of $\pi$.  With probability
one over $\pi$, the following holds simultaneously for every uniform QPT
oracle algorithm $\mathcal{A}$ and every positive polynomial $r$.

Suppose that, on input $(x,b)$ with $x\in\{0,1\}^{n}$ and
$b\in\{0,1\}$, $\mathcal{A}^{V^{\pi},\mathcal{B}}$ outputs a mixed state.
Then there is a uniform QPT oracle algorithm
\begin{align}
\Gamma
\gets
\mathcal{E}^{V^{\pi},\mathcal{B}}_{\mathcal{A},r}(x)
\label{eq:quantum-patch-extractor-complete}
\end{align}
that receives $x$ but not the bit $b$, and outputs a polynomial-length
patch $\Gamma$. (Here, $\cE$ queries $V^\pi$ only classically.)
If
\begin{align}
\rho_{b,x}
&:=
\mathcal{A}^{V^{\pi},\mathcal{B}}(x,b),
\label{eq:original-output-state-erasure-complete}
\\
\sigma^{\Gamma}_{b,x}
&:=
\mathcal{A}^{\Gamma,\mathcal{B}}(x,b),
\label{eq:patched-output-state-erasure-complete}
\end{align}
then, for all sufficiently large $n$, simultaneously for every
$x\in\{0,1\}^{n}$,
\begin{align}
\frac{1}{2}
\sum_{b\in\{0,1\}}
\mathbb{E}_{\Gamma}
\left[
\operatorname{TD}
\left(
\rho_{b,x},
\sigma^{\Gamma}_{b,x}
\right)
\right]
\leq
\frac{1}{r(n)}.
\label{eq:uniform-public-input-erasure-bound-complete}
\end{align}
The same sampled patch is used for both values of $b$.
\end{lemma}

\begin{proof}
Fix $\mathcal{A}$, $r$, and an input length $n$.  Let $T$ be the number of queries $\cA$ makes to $V^\pi$.
Set
\begin{align}
R
&:=
\left\lceil
16\,T^{2}(n+2)\,r^2n^{4}
\right\rceil,
\label{eq:erasure-stages-complete}
\\
m
&:=
\left\lceil
\log_{2}(2R)
\right\rceil,
\qquad
\text{so that}
\quad
2^{m}\geq 2R.
\label{eq:erasure-cutoff-complete}
\end{align}
Note that $R=\operatorname{poly}(n)$ and $2^{m}<4R$, so all quantities
below are polynomial in $n$.

\paragraph{The extractor.}
The extractor $\mathcal{E}$ first learns, by classical queries to
$V^{\pi}$, the complete positive table of the \emph{short} blocks,
\begin{align}
\operatorname{Short}
:=
\left\{
(\ell,u,v):
\ell<m,\
V^{\pi}(\ell,u,v)=1
\right\};
\label{eq:short-positive-table-complete}
\end{align}
this takes fewer than $2^{2m}<16R^{2}$ queries.  For a set $D$ of
triples in blocks $\ell\geq m$, write
$\Gamma(D):=\operatorname{Short}\cup D$.  

Define the auxiliary sampler
$\bar{\mathcal{A}}$ as follows: 
\begin{enumerate}
    \item 
On input a patch $\Gamma$ and $x$, sample $b\gets\{0,1\}$ and $i\gets[T]$ uniformly.
\item 
Run
$\mathcal{A}^{\Gamma,\mathcal{B}}(x,b)$ until immediately before its
$i$-th query.
\item 
Measure the query input register, and output the outcome.
\end{enumerate}

The extractor $\cE$ then proceeds as follows:
\begin{enumerate}
\item 
Initialize
$D_{0}:=\varnothing$.  
\item
For $j=0,\ldots,R-1$: run
$Q_{j}\gets\bar{\mathcal{A}}^{\Gamma(D_{j}),\mathcal{B}}(x)$; if
$Q_{j}=(\ell,u,v)$ with $\ell\geq m$ and $Q_{j}\notin D_{j}$, query
$V^{\pi}$ classically on $Q_{j}$ and set
$D_{j+1}:=D_{j}\cup\{Q_{j}\}$ if the answer is $1$ and
$D_{j+1}:=D_{j}$ otherwise; in all other cases set $D_{j+1}:=D_{j}$.
\item
Sample $J\gets\{0,\ldots,R-1\}$ and output
\begin{align}
\Gamma
:=
\Gamma(D_{J}).
\label{eq:extractor-output-patch-complete}
\end{align}
\end{enumerate}
The bits $b$ sampled inside $\bar{\mathcal{A}}$ are internal
randomness; $\mathcal{E}$ never receives the later challenge bit.  The
procedure is uniform QPT, every element of $\Gamma(D_{j})$ is a true
positive of $V^{\pi}$ by construction, and the patch $\Gamma$ has polynomial
length.

Since $\Gamma(D)$ records all positives of the short blocks and only
true positives of the long blocks, the oracles $V^{\pi}$ and
$\Gamma(D)$ differ exactly on the set of \emph{undiscovered long
positives}
\begin{align}
\Delta(D)
:=
\left\{
(\ell,u,v):
\ell\geq m,\
V^{\pi}(\ell,u,v)=1,\
(\ell,u,v)\notin D
\right\}.
\label{eq:unseen-positive-set-complete}
\end{align}
Define
\begin{align}
p_{x}(D)
&:=
\Pr\left[
\bar{\mathcal{A}}^{\Gamma(D),\mathcal{B}}(x)
\in
\Delta(D)
\right],
\label{eq:average-query-mass-complete}
\\
d_{x}(D)
&:=
\frac{1}{2}
\sum_{b\in\{0,1\}}
\operatorname{TD}
\left(
\rho_{b,x},
\sigma^{\Gamma(D)}_{b,x}
\right).
\label{eq:average-simulation-error-complete}
\end{align}
The proof consists of four steps. 

\paragraph{Step 1.}
We claim that, for every set
$D\subseteq\{(\ell,u,v):\ell\geq m,\ V^{\pi}(\ell,u,v)=1\}$,
\begin{align}
d_{x}(D)
\leq
2T\sqrt{p_{x}(D)}.
\label{eq:learn-or-simulate-average-complete}
\end{align}
This is shown based on the argument of \cite{BBBV97}.
Fix $b$, and for $i\in\{1,\ldots,T\}$ let $q_{i,b}$ be the probability
that measuring the query input register of
$\mathcal{A}^{\Gamma(D),\mathcal{B}}(x,b)$ immediately before its
$i$-th query yields a point of $\Delta(D)$.  For
$i\in\{0,\ldots,T\}$, consider the hybrid in which the first $i$
queries use $\Gamma(D)$ and the rest use $V^{\pi}$.  Adjacent
hybrids share the same state before the changed query; the two query
unitaries agree outside $\Delta(D)$, so the Euclidean distance created
at the $i$-th query is at most $2\sqrt{q_{i,b}}$, and later operations
and the partial trace do not increase it.  The triangle inequality and
Cauchy--Schwarz give
\begin{align}
\operatorname{TD}
\left(
\rho_{b,x},
\sigma^{\Gamma(D)}_{b,x}
\right)
\leq
2\sum_{i=1}^{T}\sqrt{q_{i,b}}
\leq
2T
\sqrt{
\frac{1}{T}\sum_{i=1}^{T}q_{i,b}
}.
\label{eq:learn-or-simulate-branch-complete}
\end{align}
Averaging over $b$ and using the concavity of the square root yields
Equation~\eqref{eq:learn-or-simulate-average-complete}, because the
pair $(b,i)$ sampled by $\bar{\mathcal{A}}$ is uniform and therefore
$p_{x}(D)
=\frac{1}{2}\sum_{b}\frac{1}{T}\sum_{i}q_{i,b}$.

\paragraph{Step 2.}
In this step, we show
\begin{align}
\mathbb{E}_{\mathrm{Ext},J}
\left[
d_{x}(D_{J})
\mid
\pi
\right]
\leq
2T
\sqrt{
\frac{\kappa_{\pi}(x)}{R}
}.
\label{eq:random-stage-simulation-error-saisyo}
\end{align}
Here,
\begin{align}
\kappa_\pi(x)=\mathbb{E}_{\mathrm{Ext}}[K\mid\pi], 
\label{eq:expected-total-number-of-discoveries}
\end{align}
$K\coloneqq\sum_{j=0}^{R-1}X_j$,
$X_{j}:=|D_{j+1}|-|D_{j}|\in\{0,1\}$,
and 
$\mathrm{Ext}$ denotes the internal randomness and
measurement outcomes of the extraction algorithm $\cE$.  
By construction, $X_j=1$ if and only if $Q_j\in\Delta(D_j)$, and conditioned on
$\pi$ and $D_j$, the outcome $Q_j$ is distributed as
$\bar{\mathcal{A}}^{\Gamma(D_j),\cB}(x)$. Hence
\begin{align}
\mathop{\mathbb{E}}_{\mathrm{Ext}}[X_j\mid\pi]
=\mathop{\mathbb{E}}_{\mathrm{Ext}}[p_x(D_j)\mid\pi].
\end{align}
Since $J$ is uniform over $\{0,\dots,R-1\}$ and independent of the trajectory
$(D_0,\dots,D_R)$, we obtain
\begin{align}
\mathop{\mathbb{E}}_{\mathrm{Ext},J}\left[d_x(D_J)\,\middle|\,\pi\right]
&\le 2T\,\mathop{\mathbb{E}}_{\mathrm{Ext},J}\left[\sqrt{p_x(D_J)}\,\middle|\,\pi\right]\\
&\le 2T\sqrt{\mathop{\mathbb{E}}_{\mathrm{Ext},J}\left[p_x(D_J)\,\middle|\,\pi\right]}\\
&=2T\sqrt{\frac{1}{R}\sum_{j=0}^{R-1}
\mathop{\mathbb{E}}_{\mathrm{Ext}}\left[X_j\,\middle|\,\pi\right]}\\
&=2T\sqrt{\frac{\kappa_\pi(x)}{R}},
\end{align}
where the first inequality is \cref{eq:learn-or-simulate-average-complete}
and the second is Jensen's
inequality for the concave square root.

\paragraph{Step 3.}
In this step, we show
\begin{align}
\mathbb{E}_{\pi}\left[\max_{x\in\{0,1\}^n}\kappa_\pi(x)\right]\le n+2.
\label{eq:step3goal}
\end{align}
The heart of the proof is the following exponential
moment bound, which holds uniformly for every $x\in\{0,1\}^n$:
\begin{align}
\mathop{\mathbb{E}}_{\pi,\mathrm{Ext}}\left[2^{K}\right]\le e.
\label{eq:expmoment}
\end{align}
Indeed, using $\max_x a_x\le\log_2\sum_x 2^{a_x}$ and Jensen's inequality twice
(for the concave $\log_2$ over $\pi$, and for the convex $t\mapsto 2^t$ over
$\mathrm{Ext}$), \cref{eq:expmoment} implies
\begin{align}
\mathbb{E}_{\pi}\left[\max_{x\in\{0,1\}^n}\kappa_\pi(x)\right]
\le\log_2\sum_{x\in\{0,1\}^n}\mathbb{E}_{\pi}\left[2^{\kappa_\pi(x)}\right]
\le\log_2\sum_{x\in\{0,1\}^n}\mathop{\mathbb{E}}_{\pi,\mathrm{Ext}}\left[2^{K}\right]
\le\log_2(2^n e)\le n+2.
\end{align}

It remains to prove \cref{eq:expmoment}. Let $H_j$ denote the complete classical
transcript of the extraction $\cE$ up to and including the measured outcome $Q_j$ at
stage $j$, but not the stage-$j$ test result (i.e., checking $Q_j=(\ell,u,v)$ satisfies $V^\pi(\ell,u,v)=1$). We claim that, for every $j$ and
every realization of $H_j$,
\begin{align}
\Pr_{\pi,\mathrm{Ext}}\left[X_j=1\mid H_j\right]\le\frac{1}{R}.
\label{eq:onestage}
\end{align}
Since $X_0,\dots,X_{j-1}$ are determined by $H_j$ and
$\mathbb{E}[2^{X_j}\mid H_j]=1+\Pr[X_j=1\mid H_j]$, taking conditional
expectations given $H_{R-1},H_{R-2},\dots$ in turn yields
\begin{align}
\mathop{\mathbb{E}}_{\pi,\mathrm{Ext}}\left[2^{K}\right]
\le\left(1+\frac{1}{R}\right)^{R}\le e,
\end{align}
which is \cref{eq:expmoment}.

To prove \cref{eq:onestage}, fix a realization of $H_j$ and write
$Q_j=(\ell,u,v)$. We first claim the following: if $\ell\ge m$ and the oracle
value $V^\pi(\ell,u,v)$ is not determined by the transcript, then $X_j=1$ if
and only if $\pi_\ell(u)=v$; in all other cases, $X_j=0$ deterministically.

To see the claim, we consider each case in turn. 
\begin{enumerate}
\item 
If $\ell\ge m$ and the value
is undetermined, then in particular $(\ell,u,v)\notin D_j$, so the extractor
tests $(\ell,u,v)$ and adds it to the patch exactly when $V^\pi(\ell,u,v)=1$,
i.e., when $\pi_\ell(u)=v$.
\item 
If $\ell<m$, no test is performed and $X_j=0$. 
\item 
When $\ell\ge m$ but the value is determined by
the transcript. This happens in exactly two ways. First, $(\ell,u,v)$ itself
was tested at an earlier stage: if the test was positive, then
$(\ell,u,v)\in D_j$ and no new test is performed; if it was negative, then the
new test again returns $0$. Second, an earlier positive test revealed
$\pi_\ell(u)=v'$ for some $v'\ne v$ or $\pi_\ell(u')=v$ for some $u'\ne u$;
since $\pi_\ell$ is a permutation, this forces $V^\pi(\ell,u,v)=0$, so the
test returns $0$. In every case $X_j=0$, which proves the claim.
\end{enumerate}

By the claim, it suffices to assume that $\ell\ge m$ and that the value is
undetermined, and to bound the conditional probability of $\pi_\ell(u)=v$.

We claim that, conditioned on $H_j$, the permutation $\pi_\ell$ is uniformly
distributed over the permutations of $\{0,1\}^\ell$ consistent with the
outcomes of the earlier classical tests in block $\ell$. Indeed, the
extraction interacts with $V^\pi$ only through the short-block queries, which
concern blocks $\ell'<m\le\ell$, and through the classical tests: the runs of
$\bar{\mathcal{A}}$ query only the patched oracles, which are functions of the
transcript. Since the blocks of $\pi$ are independent, no other part of the
transcript carries information about $\pi_\ell$, and the claim follows.

Among the earlier classical tests in block $\ell$, let $k_+$ be the number of
positive tests and $k_-$ that of negative tests, so that $k_++k_-\le j<R$. The
positive tests pin down the value of $\pi_\ell$ on $k_+$ inputs; we call the
remaining $2^\ell-k_+$ inputs and outputs \emph{free}. Each negative test
yields a \emph{refuted pair} $(a,b)$, meaning that $\pi_\ell(a)\ne b$; there
are $k_-$ of them. We say that a permutation $\sigma$ \emph{hits} a pair
$(a,b)$ if $\sigma(a)=b$. Let $\Sigma$ be the set of permutations of
$\{0,1\}^\ell$ that agree with the $k_+$ pinned values. In this terminology,
the claim says that the conditional distribution of $\pi_\ell$ given $H_j$ is
that of a uniform $\sigma\leftarrow\Sigma$ conditioned on hitting no refuted
pair. Hence
\begin{align}
\Pr_{\pi,\mathrm{Ext}}\left[X_j=1\,\middle|\,H_j\right]
=\Pr_{\sigma\leftarrow\Sigma}\left[\sigma(u)=v\,\middle|\,
\sigma\text{ hits no refuted pair}\right]
\le\frac{\Pr_{\sigma\leftarrow\Sigma}\left[\sigma(u)=v\right]}
{\Pr_{\sigma\leftarrow\Sigma}\left[\sigma\text{ hits no refuted pair}\right]}.
\label{eq:ratio}
\end{align}
It therefore suffices to
bound the numerator from above and the denominator from below.

For a uniform $\sigma\leftarrow\Sigma$, the values of $\sigma$ on the free
inputs form a uniformly random bijection onto the free outputs. Since the
value of $(\ell,u,v)$ is undetermined, both $u$ and $v$ are free (if $u$ were
pinned, the transcript would reveal $\pi_\ell(u)$ and thereby determine
$V^\pi(\ell,u,v)$, and similarly for $v$), and hence
\begin{align}
\Pr_{\sigma\leftarrow\Sigma}\left[\sigma(u)=v\right]=\frac{1}{2^\ell-k_+}.
\label{eq:numerator}
\end{align}
Next, each refuted pair $(a,b)$ is hit with probability at most
$1/(2^\ell-k_+)$: if both $a$ and $b$ are free this probability is exactly
$1/(2^\ell-k_+)$; if $a$ is pinned then $\sigma(a)$ equals its pinned value
$\pi_\ell(a)\ne b$, and if $b$ is pinned then $b$ is taken by its pinned input
under every $\sigma\in\Sigma$, so in either case the pair is never hit. By the
union bound,
\begin{align}
\Pr_{\sigma\leftarrow\Sigma}\left[\sigma\text{ hits no refuted pair}\right]
\ge 1-\frac{k_-}{2^\ell-k_+}>0,
\label{eq:denominator}
\end{align}
where the positivity follows from $2^\ell\ge 2^m\ge 2R>k_++k_-$. Combining
\cref{eq:ratio,eq:numerator,eq:denominator} yields
\begin{align}
\Pr_{\pi,\mathrm{Ext}}\left[X_j=1\,\middle|\,H_j\right]
\le\frac{1}{2^\ell-k_+-k_-}\le\frac{1}{2^m-R}\le\frac{1}{R},
\end{align}
where the last two inequalities use $\ell\ge m$, $k_++k_-<R$, and $2^m\ge 2R$.

\paragraph{Step 4.}
Define the worst simulation error for a fixed permutation family $\pi$
by
\begin{align}
e_{n}(\pi)
:=
\max_{x\in\{0,1\}^{n}}
\mathbb{E}_{\mathrm{Ext},J}
\left[
d_{x}(D_{J})
\mid
\pi
\right].
\label{eq:worst-public-input-error}
\end{align}
By Equation~\eqref{eq:random-stage-simulation-error-saisyo}, taking the
expectation over $\pi$, applying Jensen's inequality to the concave
square root, and using \cref{eq:step3goal}
\begin{align}
\mathbb{E}_{\pi}
\left[
e_{n}(\pi)
\right]
\leq
2T
\sqrt{
\frac{
\mathbb{E}_{\pi}
\left[
\max_{x\in\{0,1\}^{n}}\kappa_{\pi}(x)
\right]
}{R}
}
\leq
2T
\sqrt{
\frac{n+2}{R}
}
\leq
\frac{1}{2r(n)n^{2}},
\label{eq:average-worst-public-input-error}
\end{align}
where the final inequality follows from
Equation~\eqref{eq:erasure-stages-complete}.  By Markov's inequality,
\begin{align}
\Pr_{\pi}
\left[
e_{n}(\pi)
>
\frac{1}{r(n)}
\right]
\leq
r(n)\,
\mathbb{E}_{\pi}
\left[
e_{n}(\pi)
\right]
\leq
\frac{1}{2n^{2}},
\label{eq:bad-erasure-probability}
\end{align}
which is summable in $n$.  Hence, by the first Borel--Cantelli lemma,
for the fixed pair $(\mathcal{A},r)$, with probability one over $\pi$
there exists $n_{0}=n_{0}(\pi,\mathcal{A},r)$ such that
$e_{n}(\pi)\leq 1/r(n)$ for every $n\geq n_{0}$; that is,
Equation~\eqref{eq:uniform-public-input-erasure-bound-complete} holds
for every $n\geq n_{0}$ and every $x\in\{0,1\}^{n}$.  There are only
countably many pairs of a uniform oracle machine and a positive
polynomial with finite descriptions; taking the countable intersection
of the corresponding probability-one events completes the proof.
\end{proof}

\section{Proof of \cref{thm:classicaloracle}}
\label{subsec:proof-main-classical-oracle}

\begin{proof}
Apply Corollary~\ref{cor:fixed-oracle-inversion-hardness} with the side
oracle $\operatorname{QBF}$, and apply
Lemma~\ref{lem:uniform-public-input-erasure-complete} with the side
oracle $\operatorname{QBF}$.  Each application gives a probability-one
set of permutation families.  Fix one permutation family $\pi$ in the
intersection of these probability-one sets, and define
\begin{align}
\mathcal{O}
&:=
\left(
V^{\pi},
\operatorname{QBF}
\right).
\label{eq:direct-classical-oracle}
\end{align}

Items~\textup{(1)} and~\textup{(2)} follow from
Lemma~\ref{lem:decision-and-search-hardness}. 

\paragraph{Proof of Item (3).}
Let $\{D_x\}_{x\in\bit^*}$ be a sampling problem such that
$\{D_x\}_{x\in\bit^*}\in\mathsf{SampBQP}^{\mathcal{O}}$.  Then, by the
definition of $\mathsf{SampBQP}$, there exists a uniform QPT oracle
algorithm $\mathcal{Q}$ such that for all $x$ and all $\epsilon>0$,
\begin{align}
\SD\left(
\mathcal{Q}^{\mathcal{O}}\left(x,1^{\lfloor 1/\epsilon\rfloor}\right),
D_x
\right)
\leq
\epsilon.
\label{eq:item3-quantum-sampler-guarantee}
\end{align}
We construct a uniform PPT oracle algorithm $\mathcal{C}$ such that
$\mathcal{C}^{\mathcal{O}}$ satisfies the analogous guarantee.

To apply \cref{lem:uniform-public-input-erasure-complete}, define a
uniform QPT oracle algorithm $\mathcal{A}_{\mathcal{Q}}$ as follows.
Its classical input $z$ is an encoding of a pair $(x,1^{K})$, where the
encoding is chosen so that $N:=|z|\geq K$, and it also receives a bit
$b\in\{0,1\}$.  It ignores $b$ and runs the sampler $\mathcal{Q}$:
\begin{align}
\mathcal{A}_{\mathcal{Q}}^{V^{\pi},\operatorname{QBF}}
\left(\left\langle x,1^{K}\right\rangle,b\right)
&:=
\mathcal{Q}^{\mathcal{O}}\left(x,1^{K}\right).
\label{eq:item3-auxiliary-algorithm}
\end{align}
For $z=\langle x,1^{K}\rangle$ and a patch $\gamma$, define
\begin{align}
P_z
&\coloneqq
\text{the output probability distribution of }
\mathcal{Q}^{V^{\pi},\operatorname{QBF}}\left(x,1^{K}\right),\\
P_{z,\gamma}
&\coloneqq
\text{the output probability distribution of }
\mathcal{Q}^{\gamma,\operatorname{QBF}}\left(x,1^{K}\right).
\label{eq:item3-output-laws}
\end{align}
Apply \cref{lem:uniform-public-input-erasure-complete} to
$\mathcal{A}_{\mathcal{Q}}$ with the accuracy polynomial $r(N):=N$, and
let $\nu_z$ denote the output probability distribution of the patch
sampled from
$\mathcal{E}^{V^{\pi},\operatorname{QBF}}_{\mathcal{A}_{\mathcal{Q}},r}(z)$.
Since the outputs are classical strings, trace distance coincides with
statistical distance, and since $\mathcal{A}_{\mathcal{Q}}$ ignores its
bit $b$, the $b$-averaged bound of the lemma reads: for all
sufficiently large $N$ and simultaneously for every $z\in\{0,1\}^{N}$,
\begin{align}
\mathbb{E}_{\gamma\leftarrow\nu_z}
\left[
\operatorname{SD}\left(
P_z,
P_{z,\gamma}
\right)
\right]
&\leq
\frac{1}{N}.
\label{eq:item3-erasure-error}
\end{align}

We next classically simulate both random objects appearing in this
construction.  First, the procedure producing
$\gamma\leftarrow\nu_z$ can be simulated by a uniform
$\operatorname{PPT}^{V^{\pi},\operatorname{QBF}}$ algorithm. This is because
the procedure consists of classically sampling random bits, classically querying $V^\pi$, 
and running a $\mathsf{QPT}^{\rm QBF}$ algorithm, but
a $\mathsf{QPT}^{\rm QBF}$ algorithm can be simulated with 
a $\mathsf{PPT}^{\rm QBF}$ algorithm within an exponentially small precision.
Hence
there is a distribution $\widehat{\nu}_z$ that can be sampled by $\mathsf{PPT}^{V^\pi,\operatorname{QBF}}$
such that
\begin{align}
\operatorname{SD}\left(\nu_z,\widehat{\nu}_z\right)
&\leq
\negl(N).
\label{eq:item3-patch-simulation-error}
\end{align}
Second, if we know a patch $\gamma$, the circuit
$\mathcal{Q}^{\gamma,\operatorname{QBF}}(x,1^{K})$ can be
simulated with a PPT${}^{\operatorname{QBF}}$ algorithm. 
Hence
there is a distribution $\widehat{P}_{z,\gamma}$ that can be sampled by $\mathsf{PPT}^{\operatorname{QBF}}$ (with knowing $\gamma$)
such that
\begin{align}
\operatorname{SD}\left(P_{z,\gamma},\widehat{P}_{z,\gamma}\right)
&\leq
\negl(N).
\label{eq:item3-patched-output-simulation-error}
\end{align}

Let $N_{0}$ be a threshold beyond which
Equation~\eqref{eq:item3-erasure-error} holds.  Since the permutation
family $\pi$ and the algorithm $\mathcal{Q}$ are fixed in this proof,
$N_{0}$ is a fixed constant and may be hardwired into $\mathcal{C}$.

The classical sampler $\mathcal{C}^{\mathcal{O}}$ operates as follows.
On input $(x,1^{k})$, it sets
\begin{align}
K
&:=
12(k+1)+N_{0},
\label{eq:item3-internal-accuracy}
\\
z
&:=
\left\langle x,1^{K}\right\rangle,
\qquad
N:=|z|.
\label{eq:item3-expanded-input}
\end{align}
It samples a patch $\gamma\leftarrow\widehat{\nu}_z$ using the
classical patch simulator and then samples its output according to
$\widehat{P}_{z,\gamma}$.  This defines a single uniform
$\operatorname{PPT}^{\mathcal{O}}$ algorithm because $K=O(k)$ up to the
fixed additive constant $N_{0}$, and all patch descriptions and all
simulated circuits have size polynomial in $N$.

We now bound its error.  
By convexity of statistical distance and
\cref{eq:item3-erasure-error},
\begin{align}
\operatorname{SD}\left(P_z,
\mathbb{E}_{\gamma\leftarrow\nu_z}
\left[P_{z,\gamma}\right]
\right)
&\leq
\mathbb{E}_{\gamma\leftarrow\nu_z}
\left[
\operatorname{SD}\left(P_z,P_{z,\gamma}\right)
\right]
\leq
\frac{1}{N}.
\label{eq:item3-first-comparison}
\end{align}
By
\cref{eq:item3-patch-simulation-error},
\begin{align}
\operatorname{SD}\left(
\mathbb{E}_{\gamma\leftarrow\nu_z}
\left[P_{z,\gamma}\right],
\mathbb{E}_{\gamma\leftarrow\widehat{\nu}_z}
\left[P_{z,\gamma}\right]
\right)
&\leq
\operatorname{SD}\left(\nu_z,\widehat{\nu}_z\right)
\leq
\negl(N).
\label{eq:item3-second-comparison}
\end{align}
Finally, averaging
Equation~\eqref{eq:item3-patched-output-simulation-error} over
$\gamma\leftarrow\widehat{\nu}_z$ gives
\begin{align}
\operatorname{SD}\left(
\mathbb{E}_{\gamma\leftarrow\widehat{\nu}_z}
\left[P_{z,\gamma}\right],
\mathcal{C}^{\mathcal{O}}(x,1^{k})
\right)
&\leq
\negl(N).
\label{eq:item3-third-comparison}
\end{align}
Combining
Equations~\eqref{eq:item3-first-comparison}--\eqref{eq:item3-third-comparison},
we obtain
\begin{align}
\operatorname{SD}\left(
\mathcal{C}^{\mathcal{O}}(x,1^{k}),
P_z
\right)
&\leq
\frac{2}{N}.
\label{eq:item3-classical-versus-quantum}
\end{align}

By
Equation~\eqref{eq:item3-quantum-sampler-guarantee}, applied with
accuracy $1/K$,
\begin{align}
\operatorname{SD}\left(P_z,D_x\right)
&\leq
\frac{1}{K}.
\label{eq:item3-internal-quantum-error}
\end{align}
Fix any $0<\varepsilon\leq 1$ and set $k:=\lfloor 1/\varepsilon\rfloor$.  
Moreover, $N\geq K$, and therefore the triangle inequality,
Equations~\eqref{eq:item3-classical-versus-quantum} and
\eqref{eq:item3-internal-quantum-error}, and the definition of $K$ give
\begin{align}
\operatorname{SD}\left(
\mathcal{C}^{\mathcal{O}}
\left(x,1^{\lfloor 1/\varepsilon\rfloor}\right)
,
D_x
\right)
&\leq
\frac{2}{N}+\frac{1}{K}
\leq
\frac{3}{K}
\leq
\frac{1}{4(k+1)}
<
\varepsilon.
\label{eq:item3-final-sampling-error}
\end{align}
Thus 
\begin{align}
\mathsf{SampBQP}^{\mathcal{O}}
&\subseteq
\mathsf{SampBPP}^{\mathcal{O}}.
\label{eq:item3-first-inclusion}
\end{align}

\paragraph{Proof of Item (4).}
Let $\Samp^\cO$ be any candidate classically-secure
auxiliary-input 
distributional OWPuzz relative to $\mathcal{O}$.
We construct a
single uniform PPT oracle adversary $\mathcal{A}$ that breaks it.

Fix an auxiliary input $x$ and write $n:=|x|$.  
Define
\begin{align}
\mu_x
&\coloneqq
\text{the output probability distribution of }
\Samp^{V^{\pi},\operatorname{QBF}}(x),\\
\mu_x^{\gamma}
&\coloneqq
\text{the output probability distribution of }
\Samp^{\gamma,\operatorname{QBF}}(x)
\label{eq:item4-output-laws}
\end{align}
for a patch $\gamma$. 

Apply \cref{lem:uniform-public-input-erasure-complete} to the algorithm
$(x,b)\mapsto\Samp^{V^{\pi},\operatorname{QBF}}(x)$, which ignores $b$,
with the accuracy polynomial $r(n):=4p(n)$, and let $\nu_x$ denote the output probability distribution of
the patch sampled from
$\mathcal{E}^{V^{\pi},\operatorname{QBF}}_{\Samp,r}(x)$.  Since the
outputs are classical, the lemma gives, for every sufficiently large
$n$ and every $x\in\{0,1\}^{n}$,
\begin{align}
\mathbb{E}_{\gamma\leftarrow\nu_x}
\left[
\operatorname{SD}\left(
\mu_x,
\mu_x^{\gamma}
\right)
\right]
&\leq
\frac{1}{4p(n)}.
\label{eq:item4-erasure-error}
\end{align}
Moreover, exactly as in the proof of Item~\textup{(3)}
(Equation~\eqref{eq:item3-patch-simulation-error}), there is a uniform
$\operatorname{PPT}^{\mathcal{O}}$ algorithm that samples a patch from
a distribution $\widehat{\nu}_x$ satisfying
\begin{align}
\operatorname{SD}\left(\nu_x,\widehat{\nu}_x\right)
&\leq
\negl(n).
\label{eq:item4-patch-simulation-error}
\end{align}

Because $\Samp^{\gamma,\operatorname{QBF}}(x)$ is a $\mathsf{QPT}$ algorithm querying QBF,
there is a uniform PPT${}^{\operatorname{QBF}}$ algorithm
$\mathsf{Cond}$ such that, for every $n$, every $x\in\{0,1\}^{n}$, and
every patch $\gamma$ of length polynomial in $n$,
\begin{align}
\operatorname{SD}
\left(
\left(\puzz,\mathsf{Cond}(x,\gamma,\puzz)\right)_{(\puzz,\ans)\leftarrow\mu_x^{\gamma}},
\;
\mu_x^{\gamma}
\right)
\leq
2^{-n}.
\label{eq:item4-conditional-sampler}
\end{align}

The adversary $\mathcal{A}_p^{\cO}(x,\mathsf{puzz})$ operates as follows: it
samples $\gamma\leftarrow\widehat{\nu}_x$,
runs
$\widehat{\mathsf{ans}}\leftarrow\mathsf{Cond}(x,\gamma,\mathsf{puzz})$,
and outputs $\widehat{\mathsf{ans}}$. This is a single uniform
$\mathrm{PPT}^{\cO}$ algorithm.

Fix a patch $\gamma$ and let $J_{\gamma}$ denote the distribution of
$(\mathsf{puzz},\widehat{\mathsf{ans}})$ when
$(\mathsf{puzz},\mathsf{ans})\leftarrow\mu_x$ and
$\widehat{\mathsf{ans}}\leftarrow\mathsf{Cond}(x,\gamma,\mathsf{puzz})$.
Let
$\widetilde{J}_{\gamma}$ denote the same distribution but with
$(\mathsf{puzz},\mathsf{ans})\leftarrow\mu_x^{\gamma}$ instead. 
Then
$\mathrm{SD}(J_{\gamma},\widetilde{J}_{\gamma})
\le\mathrm{SD}(\mu_x,\mu_x^{\gamma})$.
Combining with \cref{eq:item4-conditional-sampler} and the triangle inequality,
\begin{equation}\label{eq:pointwise-gamma}
\mathrm{SD}\bigl(J_{\gamma},\mu_x\bigr)
\;\le\;
\mathrm{SD}\bigl(J_{\gamma},\widetilde{J}_{\gamma}\bigr)
+\mathrm{SD}\bigl(\widetilde{J}_{\gamma},\mu_x^{\gamma}\bigr)
+\mathrm{SD}\bigl(\mu_x^{\gamma},\mu_x\bigr)
\;\le\;
2\,\mathrm{SD}\bigl(\mu_x,\mu_x^{\gamma}\bigr)+2^{-n}.
\end{equation}
By
convexity of statistical distance and Equation~\eqref{eq:pointwise-gamma},
\begin{align}\label{eq:final-dist}
\mathrm{SD}\Bigl(
\bigl\{(\mathsf{puzz},\mathcal{A}_p^{\cO}(x,\mathsf{puzz}))\bigr\}
_{(\mathsf{puzz},\mathsf{ans})\leftarrow\mathsf{Samp}^{\cO}(x)}
\,,\;\mu_x\Bigr)
&\;\le\;
2\,\mathbb{E}_{\gamma\leftarrow\widehat{\nu}_x}
\bigl[\mathrm{SD}\bigl(\mu_x,\mu_x^{\gamma}\bigr)\bigr]+2^{-n}\\
&\;\le\;
2\Bigl(
\mathbb{E}_{\gamma\leftarrow\nu_x}
\bigl[\mathrm{SD}\bigl(\mu_x,\mu_x^{\gamma}\bigr)\bigr]
+\mathrm{SD}\bigl(\nu_x,\widehat{\nu}_x\bigr)
\Bigr)+2^{-n}\\
&\;\le\;
\frac{1}{2p(n)}+\negl(n)
\;<\;\frac{1}{p(n)}
\end{align}
for every sufficiently large $n$ and every $x\in\{0,1\}^{n}$. 
\end{proof}

\section{Proof of Theorem~\ref{thm:main-quantum-oracle}}
\label{subsec:proof-main-quantum-oracle}
Apply Corollary~\ref{cor:fixed-oracle-inversion-hardness} with the side
oracle $\cH$, and apply
Lemma~\ref{lem:uniform-public-input-erasure-complete} with the side
oracle $\cH$.  Each application gives a probability-one
set of permutation families.  Fix one permutation family $\pi$ in the
intersection of these probability-one sets, and define
\begin{align}
\mathcal{O}
&:=
\left(
V^{\pi},
\cH
\right).
\label{eq:direct-classical-oracle}
\end{align}

Items (1) and (2) are shown in a similar way as the proof of \cref{thm:classicaloracle}.

We prove Item (3). 
Assume toward a contradiction that an auxiliary-input EFI pair (generator) $G$ exists
relative to $\cO$, and let $p$ be the positive polynomial guaranteed by \cref{def:AIEFI}.
Define
\begin{equation}
\rho_{b,x}:=G^{\cO}(x,b).
\end{equation}
Below we construct a uniform QPT oracle distinguisher $\cD$ that depends only on $G$ and $p$, and then derive a
contradiction from the infinite set $S_\cD$ that \cref{def:AIEFI} guarantees for $\cD$.

Apply \cref{lem:uniform-public-input-erasure-complete} to the algorithm $(x,b)\mapsto G^{V^{\pi},\cH}(x,b)$ with
the accuracy polynomial
$r(n)\coloneqq16p(n)$.
The lemma
provides a random classical patch
\begin{equation}
\Gamma_x\leftarrow\mathcal{E}^{V^{\pi},\cH}_{G,r}(x),
\end{equation}
sampled independently of the later challenge bit. For every realization
$\gamma$ of $\Gamma_x$, define the patched states
\begin{equation}
\sigma^{\gamma}_{b,x}:=G^{\gamma,\cH}(x,b).
\end{equation}
Then
for every sufficiently large $n$ and every $x\in\{0,1\}^{n}$, 
\begin{equation}\label{eq:patch-err}
\frac{1}{2}\sum_{b\in\{0,1\}}
\mathbb{E}_{\gamma\leftarrow\Gamma_x}
\left[\mathrm{TD}\left(\rho_{b,x},\sigma^{\gamma}_{b,x}\right)\right]
\le\frac{1}{16p(n)}.
\end{equation}

Let $j(n)$ be a polynomial upper bound on the largest level of $\cH$ queried
by $G$ on $n$-bit inputs. Such a bound exists because the level index is
encoded in unary. For fixed $(x,\gamma)$, the patched generator
$G^{\gamma,\cH}(x,b)$ makes no query to $V^{\pi}$: the patch $\gamma$ is an
explicit finite set, and each patch query is implemented by an explicit
reversible lookup circuit. The patched generator therefore uses only
levels at most $j(n)$ of $\cH$ and explicit gates, and so the two patched
generators are valid circuit descriptions in a level-$(j(n)+1)$ query to
$\cH$.

Construct a uniform QPT distinguisher $\cD$ as follows. On input $(x,\tau)$,
it samples $\gamma\leftarrow\Gamma_x$ by running
$\mathcal{E}^{V^{\pi},\cH}_{G,r}(x)$, and queries $\cH^{(j(|x|)+1)}$ on the
descriptions of the patched $G$, the classical data
$(x,\gamma)$, the answer bit initialized to 0, and the challenge state $\tau$. 
After the query, $\cD$ measures the answer bit, and outputs it. Let
$M_{x,\gamma}$ be the projector onto the positive eigenspace of
$\sigma^{\gamma}_{1,x}-\sigma^{\gamma}_{0,x}$,
and set
\begin{equation}
e_{b,x,\gamma}:=\mathrm{TD}\left(\rho_{b,x},\sigma^{\gamma}_{b,x}\right).
\end{equation}
The distinguishing advantage satisfies
\begin{align}
\Pr\left[1\leftarrow \cD^{\cO}(x,\rho_{1,x})\right]
-\Pr\left[1\leftarrow \cD^{\cO}(x,\rho_{0,x})\right]
&=\mathbb{E}_{\gamma}\,
\mathrm{Tr}\left[M_{x,\gamma}\left(\rho_{1,x}-\rho_{0,x}\right)\right]\\
&\ge\mathbb{E}_{\gamma}\left[
\mathrm{TD}\left(\sigma^{\gamma}_{0,x},\sigma^{\gamma}_{1,x}\right)
-e_{0,x,\gamma}-e_{1,x,\gamma}\right]\\
&\ge\mathrm{TD}\left(\rho_{0,x},\rho_{1,x}\right)
-2\,\mathbb{E}_{\gamma}\left[e_{0,x,\gamma}+e_{1,x,\gamma}\right]\label{eq:adv}.
\end{align}
\cref{eq:patch-err} implies
\begin{equation}\label{eq:avg-err}
\mathbb{E}_{\gamma}\left[e_{0,x,\gamma}+e_{1,x,\gamma}\right]
\le\frac{1}{8p(n)}.
\end{equation}
Combining \cref{eq:adv,eq:avg-err}, for every sufficiently large $n$
and every $x\in\{0,1\}^n$,
\begin{align}
\Pr\left[1\leftarrow \cD^{\cO}(x,\rho_{1,x})\right]
-\Pr\left[1\leftarrow \cD^{\cO}(x,\rho_{0,x})\right]
\geq \mathrm{TD}\left(\rho_{0,x},\rho_{1,x}\right)-\frac{1}{4p(n)}.
\label{eq:advallx}
\end{align}

By the security of $G$ applied to the distinguisher $\cD$, there exist an
infinite set $S_\cD\subseteq\{0,1\}^*$ and a negligible function
$\mathrm{negl}$ such that \cref{eq:EFIfar,eq:EFIind} hold for every
$x\in S_\cD$. Because there are only finitely many strings of each length,
$S_\cD$ contains strings of unbounded length. Hence, for every sufficiently
long $x\in S_\cD$, \cref{eq:advallx} and the statistical farness
\cref{eq:EFIfar} give
\begin{align}
\Pr\left[1\leftarrow \cD^{\cO}(x,\rho_{1,x})\right]
-\Pr\left[1\leftarrow \cD^{\cO}(x,\rho_{0,x})\right]
\geq \frac{1}{p(|x|)}-\frac{1}{4p(|x|)}=\frac{3}{4p(|x|)},
\end{align}
while the computational indistinguishability \cref{eq:EFIind} gives
\begin{align}
\Pr\left[1\leftarrow \cD^{\cO}(x,\rho_{1,x})\right]
-\Pr\left[1\leftarrow \cD^{\cO}(x,\rho_{0,x})\right]
\leq \mathrm{negl}(|x|).
\end{align}
Since $3/(4p(|x|))>\mathrm{negl}(|x|)$ for all sufficiently large $|x|$,
this is a contradiction.

\ifnum\anonymous=1
\else
\vspace{1cm}
\textbf{Acknowledgements.}
TM is supported by
JST CREST JPMJCR23I3,
JST Moonshot R\verb|&|D JPMJMS2061-5-1-1, 
JST FOREST, 
MEXT QLEAP, 
the Grant-in Aid for Transformative Research Areas (A) 21H05183,
and 
the Grant-in-Aid for Scientific Research (A) No.22H00522.
\fi

\textbf{AI Disclosure.}
We used 
Claude Fable 5 to assist with
writing proofs and Introduction, and ChatGPT-5.6 Sol  for proofreading.
The authors verified the correctness and originality of all content including references.

\ifnum\submission=0
\bibliographystyle{alpha} 
\else
\bibliographystyle{splncs04}
\fi
\bibliography{abbrev3,crypto,reference,text}

@string{ieee =                  {IEEE}}

@string{springer =              "Springer"}

@InProceedings{STACS:AdcCle02,
  author =       "Mark Adcock and
                  Richard Cleve",
  title =        "A Quantum {Goldreich}-{Levin} Theorem with Cryptographic Applications",
  booktitle =    "STACS 2002, 19th Annual Symposium on Theoretical Aspects of Computer Science",
  address =      "Antibes - Juan les Pins, France",
  month =        mar # "~14--16,",
  year =         2002,
  editor =       "Helmut Alt and Afonso Ferreira",
  volume =       2285,
  series =       "Lecture Notes in Computer Science",
  pages =        "323--334",
  publisher =    "Springer, Heidelberg, Germany",
  doi =          "10.1007/3-540-45841-7_26",
}

@article{BBBV97,
  author  = {Charles H. Bennett and Ethan Bernstein and
             Gilles Brassard and Umesh V. Vazirani},
  title   = {Strengths and Weaknesses of Quantum Computing},
  journal = {{SIAM} Journal on Computing},
  volume  = {26},
  number  = {5},
  pages   = {1510--1523},
  year    = {1997},
  doi     = {10.1137/S0097539796300933}
}

@inproceedings{CCC:AarChe17,
    author       = {Scott Aaronson and
		    Lijie Chen},
    title        = {
Complexity-theoretic foundations of quantum supremacy experiments
    },
    year         = {2017},
     organization={
     CCC'17: Proceedings of the 32nd Computational Complexity Conference
     }
}

@article{Aar14,
  author    = {
 Scott Aaronson 
},
  title     = {
  The equivalence of sampling and searching
  },
  journal   = {
Theory of Computing Systems
  },
  volume    = {55},
  number    = {},
  pages     = {281-298},
  year      = {2014},
}

@article{ITCS:ABK24,
  author    = {
       Scott Aaronson and Harry Buhrman and William Kretschmer
},
  title     = {
      A Qubit, a Coin, and an Advice String Walk into a Relational Problem
  },
  journal   = {
      ITCS
  },
  volume    = {},
  number    = {},
  pages     = {},
  year      = {2024},
}

@article{BreMonShe16,
  author    = {
   Michael J. Bremner and Ashley Montanaro and Dan J. Shepherd
},
  title     = {
   Average-Case Complexity Versus Approximate Simulation of Commuting Quantum Computations
  },
  journal   = {Physical Review Letters},
  volume    = {117},
  number    = {},
  pages     = {080501},
  year      = {2016},
}

@misc{ITCS:BCQ23,
      author = {Zvika Brakerski and Ran Canetti and Luowen Qian},
      title = {On the computational hardness needed for quantum cryptography},
      howpublished = {ITCS 2023},
      year = {2023},
}

@inproceedings{Impagliazzo95,
  author    = {Russell Impagliazzo},
  title     = {A Personal View of Average-Case Complexity},
  booktitle = {Proceedings of the Tenth Annual Structure in Complexity Theory Conference,
               Minneapolis, Minnesota, USA, June 19-22, 1995},
  pages     = {134--147},
  publisher = {{IEEE} Computer Society},
  doi       = {10.1109/SCT.1995.514853},
  year      = {1995},
}

@article{Kre21,
  author    = {William Kretschmer},
  title     = {Quantum pseudorandomness and classical complexity},
  journal   = {TQC 2021},
  volume    = {},
  pages     = {},
  year      = {2021},
  url       = {},
   doi = {10.4230/LIPICS.TQC.2021.2},
}

\appendix

\clearpage


\end{document}